\documentclass[a4paper,UKenglish,cleveref, autoref, thm-restate]{lipics-v2021}

\hideLIPIcs  

\title{Small Independent Sets versus Small Separator in Geometric Intersection Graphs} 

\author{Malory {Marin}}{ENS de Lyon, CNRS, Université Claude Bernard Lyon 1, LIP, UMR 5668, 69342, Lyon cedex 07, France }{malory.marin@ens-lyon.fr}{0009-0008-8253-2831}{}
\author{Rémi {Watrigant}}{Université Claude Bernard Lyon 1, ENS de Lyon, CNRS, LIP, UMR 5668, 69342, Lyon cedex 07, France }{remi.watrigant@univ-lyon1.fr}{0000-0002-6243-5910}{}

\authorrunning{M. Malory and R. Watrigant} 

\Copyright{M. Marin and R. Watrigant} 

\ccsdesc[500]{Theory of computation~Graph algorithms analysis}

\keywords{Subexponential Algorithms, Unit Disk Graphs, 2-Subcoloring, Two-Sets Cut-Uncut} 

\category{} 

\relatedversion{} 

\nolinenumbers 

\EventEditors{John Q. Open and Joan R. Access}
\EventNoEds{2}
\EventLongTitle{42nd Conference on Very Important Topics (CVIT 2016)}
\EventShortTitle{CVIT 2016}
\EventAcronym{CVIT}
\EventYear{2016}
\EventDate{December 24--27, 2016}
\EventLocation{Little Whinging, United Kingdom}
\EventLogo{}
\SeriesVolume{42}
\ArticleNo{23}

\usepackage{tikz}
\usetikzlibrary{fit,positioning,shapes,calc,decorations.pathmorphing,decorations.pathreplacing}
\usepackage{pgfplots}
\pgfplotsset{compat=newest}
\pgfplotsset{width=15cm}
\usepgfplotslibrary{colorbrewer} 
\usepackage{tikz-3dplot}

\usepackage{boxedminipage}
\usepackage[most]{tcolorbox}

\newcommand{\problemdef}[3]{
	
    \begin{tcolorbox}[
        enhanced,
        title={\large \color{black}\textsc{#1}},
        colback=white,
        boxrule=0.3pt,
        attach boxed title to top left={xshift=1cm, yshift*=-2.5mm},
        boxed title style={size=small, frame hidden, colback=white}
    ]
    \textbf{Input:} \hspace*{2mm} #2 \\
    \textbf{Output:} \hspace*{2mm} #3
    \end{tcolorbox}

}
\DeclareMathOperator{\tw}{tw}
\DeclareMathOperator{\ta}{tree-\alpha}
\DeclareMathOperator{\alw}{\alpha-lw}
\DeclareMathOperator{\amod}{\alpha-mod}

\newcommand{\torso}{\mathsf{torso}}

\newtheorem{problem}{Problem}

\begin{document}

\maketitle

\begin{abstract}

While most classical NP-hard graph problems cannot be solved in time $2^{o(n)}$ on general graphs under the Exponential Time Hypothesis (ETH), many exhibit the \emph{square-root phenomenon} and admit optimal algorithms running in time $2^{O(\sqrt{n})}$ on certain geometric intersection graphs, such as planar graphs or unit disk graphs. In 2018, de Berg \emph{et~al.} developed a general algorithmic framework for such problems on intersection graphs of similarly sized fat objects in $\mathbb{R}^d$, achieving running times of the form $2^{O(n^{1-1/d})}$, along with matching lower bounds under ETH.

In this paper, we identify problems that do not exhibit the square-root phenomenon, yet still admit subexponential algorithms on intersection graphs of similarly sized fat objects in $\mathbb{R}^d$, for every fixed dimension $d \geqslant 2$. We introduce the notion of a \emph{weak square-root phenomenon}: problems that can be solved in time $2^{\tilde{O}(n^{1-1/(d+1)})}$, and for which matching lower bounds hold under ETH. We develop both an algorithmic framework and a corresponding lower bound framework. As concrete examples, we show that the problems \textsc{2-Subcoloring} and \textsc{Two Sets Cut-Uncut} exhibit this behavior.

Our algorithms rely on a new win-win structural theorem, which can be informally stated as follows: every such graph admits a sublinear separator whose removal leaves connected components with sublinear independence number. 
To facilitate the design of these algorithms, we introduce a new graph parameter, the \emph{$\alpha$-modulator number}, which generalizes both the independence number and the vertex cover number.

\end{abstract}


\section{Introduction}

Many classical graph problems require running time $2^{\Omega(n)}$ on general graphs. However, the same problems often admit \emph{subexponential time} algorithms, that is, algorithms running in time $2^{o(n)}$, when the input graph enjoys certain geometric properties. The most famous example is the \emph{square-root phenomenon} in planar graphs: numerous problems can be solved in time $2^{O(\sqrt{n})}$ on this class. These include, among others, \textsc{Maximum Independent Set}, \textsc{Hamiltonian Cycle}, and \textsc{$3$-Coloring}.

The main structural insight enabling these algorithms is the celebrated \emph{Planar Separator Theorem}, which states that every planar graph admits a balanced separator of size $O(\sqrt{n})$. A direct consequence is that every planar graph has treewidth $O(\sqrt{n})$. Since many classical problems can be solved in time $2^{\operatorname{tw}(G)} \cdot n^{O(1)}$, it follows that on planar graphs these problems admit $2^{O(\sqrt{n})}$-time algorithms.

This naturally raises the question: \emph{which graph classes exhibit this square-root phenomenon?}
One well-studied generalization is the class of $H$-minor-free graphs (for any fixed graph $H$), which generalize planar graphs via Kuratowski-Wagner’s theorem. These classes also have treewidth $O(\sqrt{n})$ and therefore enjoy subexponential-time algorithms for a wide range of problems.

Another direction for generalizing planar graphs is through the lens of \emph{geometric intersection graphs}. Given a set of objects $F \subseteq \mathbb{R}^d$, the intersection graph $G[F]$ has one vertex per object of $F$ and edges between pairs of intersecting objects. A classical example is that of \emph{unit disk graphs}, the intersection graphs of unit-radius disks in the plane. Despite their geometric origin, such graphs behave very differently from planar graphs: objects may pairwise intersect in large numbers, and unit disk graphs can contain arbitrarily large cliques, unlike planar graphs. Consequently, general intersection graphs do not admit sublinear treewidth, and the standard treewidth-based subexponential techniques do not apply.

In this paper, we focus on geometric objects that remain "close" to $d$-dimensional balls. More precisely, given a constant $\beta \geqslant 1$, a family of objects $F$ in $\mathbb{R}^d$ is said to be \emph{similarly sized $\beta$-fat} if for every object $O \in F$ there exist two $d$-dimensional balls $B_{\mathrm{in}} \subseteq O \subset B_{\mathrm{out}}$ such that $B_{\mathrm{in}}$ has diameter $1$ and $B_{\mathrm{out}}$ has diameter $\beta$. When $F$ is a set of balls of equal radius in $\mathbb{R}^d$, the intersection graph $G[F]$ is called a \emph{$d$-dimensional unit ball graph}, or simply a \emph{unit ball graph}.

In their seminal work, de Berg \emph{et~al.}~\cite{de2018framework} developed a powerful algorithmic framework showing that intersection graphs of similarly sized fat objects in $\mathbb{R}^d$ admit $2^{O(n^{1-1/d})}$-time algorithms for a broad range of problems. They also proposed a matching lower-bound framework, establishing ETH-tightness for all the problems they considered. Their results unify a long line of previous research~\cite{alber2004geometric,fomin2019finding,fomin2012bidimensionality,marx2014limited} and provide a robust methodology covering a wide variety of problems.

However, their framework does not cover all problems. Several classical problems exhibit a \emph{curse of dimensionality}: there exists a dimension $d$ such that no subexponential-time algorithm exists. 
A prominent example is the \textsc{Maximum Clique} problem on unit ball graphs: it is polynomial-time solvable in dimension $d=2$, yet admits no subexponential-time algorithm in dimension $d=4$ under the ETH~\cite{bonamy2021eptas}. Another example is \textsc{Clique Cover}, which can be solved in time $2^{O(\sqrt{n})}$ on unit disk graphs, but has no subexponential-time algorithm in dimensions $d \geqslant 5$, under the ETH~\cite{koana2024subexponential}.

In this paper, we aim to capture problems that lie "in between" the two previously described extremes: they admit subexponential-time algorithms on intersection graphs of similarly sized fat objects in $\mathbb{R}^d$ for any fixed dimension $d \geqslant 2$, but with a slightly worse running time than the one obtained via the framework of de Berg~\emph{et~al}. We call this behavior the \emph{weak square-root phenomenon}. More precisely, our framework yields algorithms with running time 
\[
    2^{\tilde{O}\!\left(n^{1 - 1/(d+1)}\right)},
\]
where the $\tilde{O}(\cdot)$ notation hides polylogarithmic factors. We also provide nearly matching lower bounds (up to these polylogarithmic factors) under the ETH.

In particular, we show that certain \emph{binary classification} problems exhibit this intermediate behavior. We propose a general framework for obtaining (almost) ETH-tight algorithms together with corresponding lower bounds. The problems captured by our framework share several structural properties that can be informally explained through the lens of \emph{parameterized complexity}. Roughly speaking, a graph problem $\Pi$ falls within our framework if:
\begin{enumerate}
    \item It is FPT with respect to the \emph{vertex cover number} of the input graph, with a single-exponential dependency on the vertex cover number; and
    \item It is FPT or XP with respect to the \emph{independence number} $\alpha(G)$ of the input graph $G$ (which is the maximum size of an independent set), again with a single-exponential dependency in the parameter.
\end{enumerate}

The remainder of the introduction presents an overview of the algorithmic framework, including examples of problems it captures, as well as the corresponding lower-bound framework.

\paragraph*{A win/win approach}
As discussed earlier, our goal is to capture problems that are simultaneously FPT (or XP) with respect to both the vertex cover number and the independence number of the input graph. The central ingredient of our framework is a new separator theorem which, informally speaking, asserts that every intersection graph of similarly sized fat objects in $\mathbb{R}^d$ has a small subset of vertices whose removal leaves each connected component with a small independence number.

\begin{theorem}\label{thm:Separator} Let $d\geqslant 2$ and $\beta \geqslant 1$ be two constants, and let $G$ be the intersection graph of $n$ similarly sized $\beta$-fat objects in $\mathbb{R}^d$. There exists a subset $S\subseteq V(G)$ of vertices such that:
\begin{itemize}
    \item $|S| \leqslant d n^{1-1/(d+1)}$, and
    \item for all connected component $C$ of $G-S$, $\alpha(G[C])\leqslant (2\beta d)^d n^{1-1/(d+1)}$.
\end{itemize}
Moreover, when the geometric representation is given, such a set $S$ can be computed in polynomial time.
\end{theorem}

This separator theorem naturally gives rise to a new graph parameter. For a graph $G$, the \emph{$\alpha$-modulator number} is the smallest integer $k$ such that there exists a set $S \subseteq V(G)$ of size at most $k$ for which every connected component of $G - S$ has independence number at most $k$.

From Theorem~\ref{thm:Separator}, it follows that intersection graphs of similarly sized $\beta$-fat objects in $\mathbb{R}^d$ has $\alpha$-modulator number $O(n^{1 - 1/(d+1)})$.

On graphs of bounded $\alpha$-modulator number, the algorithmic strategy is as follows:
\begin{itemize}
    \item Solve the problem independently on each connected component of $G - S$, using an algorithm that is \textsc{FPT} or \textsc{XP} parameterized by the independence number, for every partial solution on $S$;
    \item Combine these partial solutions to obtain a global solution.
\end{itemize}

This immediately yields $2^{\tilde{O}(n^{1 - 1/(d+1)})}$-time algorithms for geometric intersection graphs.

We also identify two natural problems that fit within this framework, for which we provide matching lower bounds under the ETH.

\paragraph*{$2$-subcoloring} A $k$-subcoloring of a graph is an assignment of colors to the vertices of the graph such that each color class induces a disjoint union of complete graphs. The \emph{subchromatic number} of a graph is the minimum $k$ such that the graph has a $k$-subcoloring. This notion was introduced in 1989 by Albertson \emph{et~al.}~\cite{albertson1989subchromatic} as a natural generalization of the chromatic number, which may be small even for dense graphs such as complete graphs. Note that disjoint unions of cliques are sometimes referred to as \emph{cluster graphs}, and this class is equivalent to the class of $P_3$-free graphs, that is, graphs with no induced path on three vertices.

\problemdef{2-Subcoloring}{Graph $G$}{Does $G$ has a $2$-subcoloring, \emph{i.e.} a partition $(A,B)$ of $V(G)$ such that $G[A]$ and $G[B]$ are both a disjoint union of cliques?}

The special case $k=2$ is already of particular interest. Indeed, the \textsc{2-Subcoloring} problem is NP-complete~\cite{gimbel}, even on planar graphs~\cite{fiala,ochem} and on unit disk graphs~\cite{marin2025subcoloring}. On the contrary, the problem is polynomial-time solvable on interval graphs~\cite{fiala} and on graphs of bounded independence number~\cite{kanj2018parameterized}. More precisely, it was shown in~\cite{kanj2018parameterized} that \textsc{2-Subcoloring} is FPT when parameterized by the number of clusters in the $2$-subcoloring, which is a lower bound on the independence number of the input graph. In~\cite{fiala}, an FPT algorithm parameterized by the treewidth of the input graph was provided that solves \textsc{2-Subcoloring}.

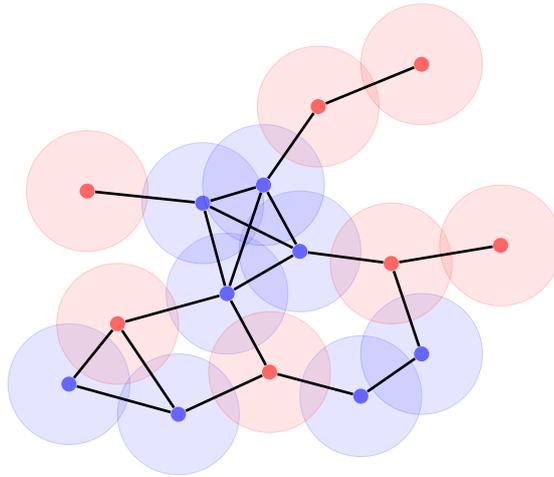
\begin{figure}
\centering
\begin{tikzpicture}[scale=0.8]
\def\r{1}

\coordinate (v1)  at (2.2,2);
\coordinate (v2)  at (4,1.5);
\coordinate (v3)  at (5.5,2.2);
\coordinate (v4)  at (7,1.8);
\coordinate (v5)  at (8,2.5);

\coordinate (v6)  at (3,3);
\coordinate (v7)  at (4.8,3.5);
\coordinate (v8)  at (6,4.2);
\coordinate (v9)  at (7.5,4);
\coordinate (v10) at (9.3,4.3);

\coordinate (v11) at (2.5,5.2);
\coordinate (v12) at (4.4,5);
\coordinate (v13) at (5.4,5.3);
\coordinate (v14) at (6.3,6.6);
\coordinate (v15) at (8,7.3);

\foreach \i in {1,2,4,5,7,8,12,13}{
    \draw[blue, fill=blue!50, opacity = 0.2] (v\i) circle (\r);
    \node[circle, fill=blue!60, inner sep=2pt] (u\i) at (v\i) {};
}
\foreach \i in {3,6,9,10,11,14,15}{
    \draw[red, fill=red!50, opacity = 0.2] (v\i) circle (\r);
    \node[circle, fill=red!60, inner sep=2pt] (u\i) at (v\i) {};
}

\draw[line width = 1] (u1) -- (u2) -- (u6) -- (u1);
\draw[line width = 1] (u2) -- (u3);
\draw[line width = 1] (u3) -- (u4);
\draw[line width = 1] (u4) -- (u5);

\draw[line width = 1] (u6) -- (u7);
\draw[line width = 1] (u7) -- (u8);
\draw[line width = 1] (u8) -- (u9);
\draw[line width = 1] (u9) -- (u10);

\draw[line width = 1] (u11) -- (u12);
\draw[line width = 1] (u12) -- (u13);
\draw[line width = 1] (u13) -- (u14);
\draw[line width = 1] (u14) -- (u15);

\draw[line width = 1] (u3) -- (u7);
\draw[line width = 1] (u5) -- (u9);

\draw[line width = 1] (u7) -- (u12) -- (u8) -- (u13) -- (u7) ;
\end{tikzpicture}

\caption{A $2$-subcoloring of a unit disk graph.}\label{fig:IntroSub}
\end{figure}

\paragraph*{Two Sets Cut-Uncut and related} Given a graph $G$ and two disjoint sets of terminals $S,T\subseteq V(G)$, an $S$-$T$-cut of $G$ is partition $(A,B)$ of $V(G)$ such that $S\subseteq A$ and $T\subseteq B$. In the \textsc{Two Sets Cut-Uncut} problem, we are given an edge-weighted graph $(G,w)$, where $w:E(G) \rightarrow \mathbb{N}$ , and we aim for an $S$-$T$ cut of minimum weight respecting some connectivity constraint. The \emph{weight} of a cut $(A,B)$ is the sum of the weight of edges $uv\in E(G)$ such that $u\in A$ and $v\in B$.

\problemdef{Two Sets Cut-Uncut}{A weighted graph $(G,w)$ with $w:E(G)\rightarrow \mathbb{N}$, sets $S,T\subseteq V(G)$ with $S \cap T = \emptyset$}{An $S$-$T$-cut $(A, B)$ of minimum weight such that $S$ (resp.\ $T$) is in a connected component of $G[A]$ (resp.\ $G[B]$).}

Bentert \emph{et~al.}~\cite{bentert2023two} introduced the unweighted version of this problem. They proved that it is \textsf{W}[1]-hard when parameterized by $|T|$, even in the special case where $|S| = 1$ on general graphs. In contrast, they showed that the unweighted problem becomes fixed-parameter tractable on planar graphs when parameterized by $|S \cup T|$. They further established fixed-parameter tractability on planar graphs when parameterized by the minimum number of faces in a planar embedding such that each terminal is incident to at least one of these faces. Later, Bentert \emph{et~al.}~\cite{bentert2024parameterized} provided a comprehensive study of the problem under various structural parameterizations. In particular, they showed that the problem is FPT when parameterized by the treewidth of the input graph (and hence the vertex cover number), and XP when parameterized by the independence number.

This problem naturally generalizes two classical problems.  
The first is the \textsc{2-Disjoint Connected Subgraphs} problem, which asks only for the existence of a partition $(A,B)$ satisfying the terminal constraints and the connectivity requirements, without optimizing the number of crossing edges. This problem is already NP-complete and has been extensively studied in graph theory and computational geometry~\cite{cygan2014solving,gray2012removing,van2009partitioning,kern2022disjoint,paulusma2011partitioning,telle2013connecting}.

The second related problem is \textsc{Network Diversion}, in which $S$ and $T$ are singletons $\{s\}$ and $\{t\}$, respectively, and there is a designated edge $b = uv \in E(G)$ that must lie in the cut between $A$ and $B$. One can solve \textsc{Network Diversion} by solving \textsc{Two Sets Cut-Uncut} on the two instances $(S=\{s,u\}, T=\{t,v\})$ and $(S'=\{s,v\}, T'=\{t,u\})$. The \textsc{Network Diversion} problem has attracted considerable attention in the computer networking community~\cite{cullenbine2013theoretical,curet2001network,erken2002branch,kallemyn2015modeling,lee2019combinatorial}.

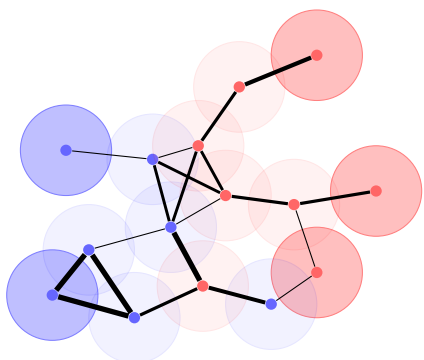
\begin{figure}
\centering
\begin{tikzpicture}[scale=0.6]
\def\r{1}

\coordinate (v1)  at (2.2,2);
\coordinate (v2)  at (4,1.5);
\coordinate (v3)  at (5.5,2.2);
\coordinate (v4)  at (7,1.8);
\coordinate (v5)  at (8,2.5);

\coordinate (v6)  at (3,3);
\coordinate (v7)  at (4.8,3.5);
\coordinate (v8)  at (6,4.2);
\coordinate (v9)  at (7.5,4);
\coordinate (v10) at (9.3,4.3);

\coordinate (v11) at (2.5,5.2);
\coordinate (v12) at (4.4,5);
\coordinate (v13) at (5.4,5.3);
\coordinate (v14) at (6.3,6.6);
\coordinate (v15) at (8,7.3);


\foreach \i in {1,11}{
    \draw[blue, fill=blue!50, opacity = 0.5] (v\i) circle (\r);
    \node[circle, fill=blue!60, inner sep=1.5pt] (u\i) at (v\i) {};
}

\foreach \i in {5,10,15}{
    \draw[red, fill=red!50, opacity = 0.5] (v\i) circle (\r);
    \node[ circle, fill=red!60, inner sep=1.5pt] (u\i) at (v\i) {};
}

\foreach \i in {2,4,6,7,12}{
    \draw[blue, fill=blue!50, opacity = 0.1] (v\i) circle (\r);
    \node[circle, fill=blue!60, inner sep=1.5pt] (u\i) at (v\i) {};
}
\foreach \i in {3,8,9,13,14}{
    \draw[red, fill=red!50, opacity = 0.1] (v\i) circle (\r);
    \node[circle, fill=red!60, inner sep=1.5pt](u\i) at (v\i) {};
}

\draw[line width = 2] (u1) -- (u2) -- (u6) -- (u1);
\draw[line width = 1.2] (u2) -- (u3);
\draw[line width = 1.5] (u3) -- (u4);
\draw (u4) -- (u5);

\draw (u6) -- (u7);
\draw (u7) -- (u8);
\draw[line width = 1.2] (u8) -- (u9);
\draw[line width = 1.2] (u9) -- (u10);

\draw (u11) -- (u12);
\draw (u12) -- (u13);
\draw[line width = 1.2] (u13) -- (u14);
\draw[line width = 1.6] (u14) -- (u15);

\draw[line width = 2] (u3) -- (u7);
\draw (u5) -- (u9);

\draw[line width = 1.1] (u7) -- (u12) -- (u8) -- (u13) -- (u7) ;

\end{tikzpicture}

\caption{An example instance of the \textsc{Two Sets Cut-Uncut} problem on a unit disk graph together with a solution $(A, B)$ represented with blue and red disks. The terminal sets $S$ and $T$ are shown as opaque blue and red disks, respectively. Edge widths are proportional to their weights.}\label{fig:IntroCutUncut}
\end{figure}

\paragraph*{Lower Bound Framework.} We establish matching lower bounds of the form $2^{o\left(n^{1-1/(d+1)}\right)}$ via reductions from \textsc{Monotone Not-All-Equal-3-SAT}. 
The reduction embeds the input formula into a $d$-dimensional grid of side length $O(n^{1/d})$ and uses a breadth-first search tree $T$ of height $O(n^{1/d})$ to control the propagation of truth assignments. 
Each clause is associated with a distinct grid cell and replaced by a gadget that satisfies the following properties:
\begin{enumerate}
    \item it can be realized within a $d$-dimensional box of constant side length;
    \item it propagates the truth values of its literals to selected adjacent cells of the grid;
    \item it enforces the Not-All-Equal constraint of the clause.
\end{enumerate}
A crucial feature of the construction is that all occurrences of any given literal are restricted to a single root-to-leaf path of $T$. 
As a consequence, the resulting graph has size $O(n^{(d+1)/d})$, which yields the claimed ETH-based lower bound.

\section{Preliminaries}\label{sec:preliminaries}

\subparagraph*{Graph notations.} 
Let $G$ be a simple graph. We denote by $V(G)$ and $E(G)$ the set of vertices and the set of edges of $G$, respectively. When there is no ambiguity, we denote by $n$ the number of vertices of $G$, and by $m$ the number of edges of $G$. An \emph{independent set} of $G$ is a set of pairwise non-adjacent vertices, and we denote by $\alpha(G)$ the \emph{independence number} of $G$, i.e., the size of a maximum independent set. Similarly, a \emph{clique} in $G$ is a set of pairwise adjacent vertices, and we denote by $\omega(G)$ the size of a maximum clique. A \emph{vertex clover} of $G$ is a subset of vertices $S\subseteq V(G)$ such that for each edge $uv\in E(G)$, either $u\in S$ or $v\in S$. The minimum size of a vertex cover of $G$, called the \emph{vertex cover number}, is denoted by $\text{vc}(G)$. Given a set $R \subseteq V(G)$, we use $G[R]$ to denote the subgraph induced by $R$, and $G - R$ to denote the graph induced by $V(G) \setminus R$. For a vertex $v \in V(G)$, we denote by $N(v)$ the \emph{open neighborhood} of $v$, that is, $N(v) = \{u \in V(G) \mid uv \in E(G)\}$, and by $N[v]$ its \emph{closed neighborhood}, defined as $N[v] = N(v) \cup \{v\}$. 

\subparagraph*{Functions and sets} Given a function $f: X \rightarrow Y$ and a subset $X'\subseteq X$, we denote by $f\restriction X'$ the \emph{restriction of $f$ to $X'$}, that is the function $h:X'\rightarrow Y$ defined by $h(x)=f(x)$ for all $x\in X'$. 


\subparagraph*{Geometric notations.} Given a constant $d\geqslant 2$ and two points $x,y\in \mathbb{R}^d$, we call the \emph{distance} between $x$ and $y$ the Euclidean distance between $x$ and $y$ in $\mathbb{R}^d$. Given two sets of points $X,Y\subset \mathbb{R}^d$, we call the \emph{distance between $X$ and $Y$} the minimum distance between a point $x\in X$ and a point $y\in Y$. For $j\in \{1,...,d\}$, we denote by $e_j$ the unit vector such that its $j$th entry equals~$1$.

\section{Separator theorem and new width parameters}\label{sec:SepWidth}

\subsection{The separator theorem}

This section is dedicated to the proof of Theorem~\ref{thm:Separator}.  
Let $d \geqslant 2$ and $\beta \geqslant 1$ be two constants. Let $F$ be a set of $n$ similarly sized $\beta$-fat objects in $\mathbb{R}^d$, and let $G := G[F]$.  
For every object $o \in F$, denote by $X_o \in \mathbb{R}^d$ the center of the smallest enclosing ball of $o$ (which has diameter at most $\beta$). 

\medskip

Let $p = \left\lceil n^{1/(d+1)} \right\rceil$ (so that $p > 1$).  
For $d' \in [d]$ and $i \in \mathbb{Z}/p\mathbb{Z}$, define
\[
\mathrm{Slab}_{d',i} := \bigcup_{\substack{q \in \mathbb{Z} \\ q \equiv i \ (\mathrm{mod}\ p)}} 
\left\{ (x_1,\dots,x_d) \in \mathbb{R}^d \;\middle|\; x_{d'} \in [\beta q, \beta(q+1)] \right\}.
\]
Informally, $\mathrm{Slab}_{d',i}$ consists of all slabs of width $\beta$ whose index is congruent to $i$ modulo $p$.  An illustration in the case of unit disks ($d=2$ and $\beta=1$) is given on Figure~\ref{fig:Slabs}.

For each $(d',i)$, define $V_{d',i} \subseteq V(G)$ as the set of vertices corresponding to objects $o \in F$ such that $X_o \in \mathrm{Slab}_{d',i}$.

\begin{figure}
    \centering
    \begin{tikzpicture}[scale=0.8]
    \foreach \i in {0,...,10}{
        \draw[gray] (\i,-0.1) -- (\i, 10.1) ;
        \node[draw=none] () at (\i, -0.5) {\i};
    }
    \foreach \i in {0,4,8}{
        \draw[draw=none, fill = blue, fill opacity = 0.4] (\i,0) rectangle (\i+1, 10);
    }
    
    \foreach \j in {0,...,10}{
        \draw[gray] (-0.1,\j) -- (10.1,\j) ;
        \node[draw=none] () at (-0.5,\j) {\j};
    }
    \foreach \j in {1,5,9}{
        \draw[draw=none, fill = red, fill opacity = 0.4] (0,\j) rectangle (10, \j+1);
    }
    
    \draw[draw=none, fill = blue, fill opacity = 0.4] (11.5, 9) rectangle (12, 9.5);
    \node[draw=none]() at (13, 9.2) {$\mathrm{Slab}_{1,0}$};
    
    \draw[draw=none, fill = red, fill opacity = 0.4] (11.5, 8) rectangle (12, 8.5);
    \node[draw=none]() at (13, 8.2) {$\mathrm{Slab}_{2,1}$};
    \foreach \x/\y in {
    0.42/7.13, 9.85/1.27, 3.64/5.92, 6.71/8.44, 2.18/0.56,
    7.39/3.21, 1.07/9.63, 4.55/2.87, 8.23/6.14, 5.90/4.02,
    0.91/1.74, 9.12/8.36, 3.33/7.58, 6.02/0.95, 2.77/3.49,
    7.84/5.61, 1.56/6.92, 4.98/9.11, 8.67/2.43, 5.21/7.79,
    0.15/4.66, 9.48/3.05, 3.87/1.28, 6.94/6.73, 2.49/8.90,
    7.08/2.16, 1.93/5.47, 4.12/0.34, 8.51/9.76, 5.73/3.88,
    0.68/6.25, 9.27/4.59, 3.05/2.71, 6.36/7.02, 2.84/1.93,
    7.61/8.58, 1.24/3.36, 4.79/6.67, 8.06/0.82, 5.47/5.13,
    0.33/8.01, 9.66/6.48, 3.58/4.20, 6.11/2.55, 2.06/7.34,
    7.92/1.64, 1.78/0.99, 4.44/8.73, 8.95/3.52, 5.02/6.89,
    0.57/2.11, 9.03/7.27, 3.21/9.48, 6.63/4.76, 2.95/5.08,
    7.14/0.67, 1.41/8.29, 4.67/3.60, 8.38/7.91, 5.88/1.45,
    0.26/5.72, 9.74/2.98, 3.76/6.35, 6.27/9.19, 2.63/4.41,
    7.53/3.77, 1.89/7.66, 4.01/1.02, 8.72/5.54, 5.36/8.87,
    0.84/3.14, 9.21/0.23, 3.47/8.55, 6.58/5.97, 2.11/6.03,
    7.26/9.32, 1.63/2.80, 4.90/7.44, 8.14/4.33, 5.64/0.51
    } {
        \draw (\x,\y) circle (0.5);
    }
    
    \draw[->] (-0.1,0) -- (10.5,0);
    \draw[->] (0, -0.1) -- (0,10.5);

    \end{tikzpicture}
    \caption{An example of the slabs defined in the proof of Theorem~\ref{thm:Separator}, in the case of $d=2$, $\beta=1$ and $p=5$.}
    \label{fig:Slabs}
\end{figure}
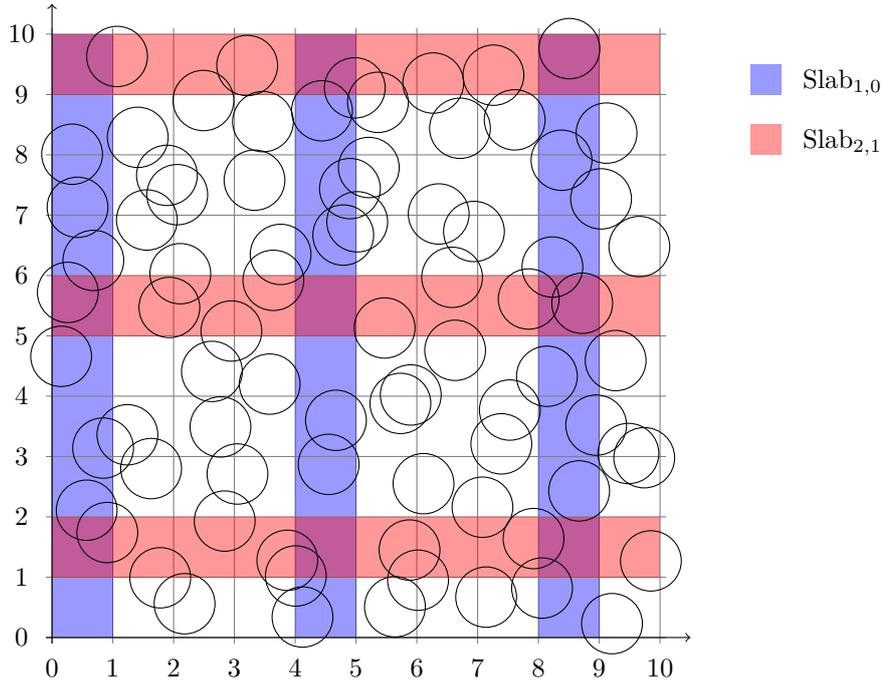

\begin{claim}\label{claim:Small}
For every $d' \in [d]$, there exists $i \in \{0,\dots,p-1\}$ such that $|V_{d',i}| \leqslant n^{1 - 1/(d+1)}$.
\end{claim}

\begin{claimproof}
For fixed $d'$, the sets $\{V_{d',i}\}_{0 \leqslant i < p}$ are pairwise disjoint and their union has size at most $n$. Hence, for some $i$, $|V_{d',i}| \leqslant \frac{n}{p} \leqslant n^{1 - 1/(d+1)}$.
\end{claimproof}

Applying Claim~\ref{claim:Small} for each $d' \in [d]$, we obtain indices $(i_1,\dots,i_d) \in \{0,\dots,p-1\}^d$ such that
\[
\left| \bigcup_{d'=1}^d V_{d',i_{d'}} \right| \leqslant d\, n^{1 - 1/(d+1)}.
\]
Let
\[
S := \bigcup_{d'=1}^d V_{d',i_{d'}}.
\]

\begin{claim}\label{claim:Hypercube}
Let $F' \subseteq F$ be the set of objects corresponding to a connected component of $G - S$.  
Then all centers $\{X_o : o \in F'\}$ lie in a hypercube of side length at most $\beta p$.
\end{claim}

\begin{claimproof}
Fix a dimension $d' \in [d]$.  
Observe that if two points belong to slabs $\mathrm{Slab}_{d',i}$ and $\mathrm{Slab}_{d',i'}$ with $|i - i'| \geqslant 2$, then their $d'$-th coordinates differ by at least $\beta$, and hence their Euclidean distance is greater than $\beta$.  

Since each object has diameter at most $\beta$, two such objects cannot intersect, and therefore the corresponding vertices are not adjacent in $G$.

Now consider a connected component $C$ of $G - S$, and let $F'$ be its corresponding objects.  
Since all vertices whose centers lie in $\mathrm{Slab}_{d',i_{d'}}$ have been removed, the remaining centers cannot cross this slab. Hence, along each dimension $d'$, the centers of objects in $F'$ must lie within at most $p$ consecutive slabs.

Since each slab has width $\beta$, this implies that along each coordinate the range is at most $\beta p$. Therefore, all centers lie in a hypercube of side length at most $\beta p$.
\end{claimproof}

\begin{claim}\label{claim:AlphaCube}
Let $F' \subseteq F$ be a set of objects contained in a hypercube $H \subseteq \mathbb{R}^d$ of side length $a$.  
Then $\alpha(G[F']) \leqslant (ad)^d$.
\end{claim}

\begin{claimproof}
Let $M := \alpha(G[F'])$.  
By the definition of $\beta$-fat objects (after scaling), each object contains a ball of radius at least $1/2$. Therefore, the $M$ objects of a maximum independent set contain $M$ pairwise disjoint balls of radius $1/2$.

Let $V_d$ denote the volume of a $d$-dimensional ball of radius $1/2$. One can show that $V_d \geqslant \left(\frac{1}{d}\right)^d$. Since all objects lie in $H$, the hypercube of volume $a^d$ contains $M$ disjoint balls, hence $M \cdot V_d \leqslant a^d$. Using the bound on $V_d$, we obtain $M \leqslant a^d \cdot d^d = (ad)^d$.
\end{claimproof}

Let $C$ be a connected component of $G - S$, and let $F'$ be its corresponding objects.  
By the Claim~\ref{claim:Hypercube}, all centers lie in a hypercube of side length $\beta p$. Since each object has diameter at most $\beta$, all objects are contained in a hypercube of side length $\beta(p+1)$.

By Claim~\ref{claim:AlphaCube}, we obtain
\[
\alpha(G[C]) \leqslant (\beta(p+1)d)^d.
\]
Since $p = \lceil n^{1/(d+1)} \rceil$, this yields
\[
\alpha(G[C]) \leqslant (2 \beta d)^d \, n^{1 - 1/(d+1)},
\]
which concludes the proof.

\subsection{An adapted width parameter}

Theorem~\ref{thm:Separator} motivates the definition of a new graph parameter, which we call the \emph{$\alpha$-modulator number}.

\begin{definition}
Given a graph $G$, the \emph{$\alpha$-modulator number of $G$}, denoted $\amod(G)$, is the minimum integer $k$ such that there exists a subset $S \subseteq V(G)$ satisfying:
\begin{itemize}
    \item $|S| \leqslant k$;
    \item for every connected component $C$ of $G - S$, $\alpha(G[C]) \leqslant k$.
\end{itemize}
The subset $S$ is called the \emph{modulator}.
\end{definition}

Observe that both the vertex cover number and the independence number are natural upper bounds for this parameter. Using this definition, Theorem~\ref{thm:Separator} can be reformulated as follows.

\begin{theorem}[Reformulation of Theorem~\ref{thm:Separator}]\label{thm:BoundAlphaMod}
Let $d \geqslant 2$ and $\beta \geqslant 1$ be constants.  
Any intersection graph $G$ of $n$ similarly sized $\beta$-fat objects in $\mathbb{R}^d$ satisfies
\[
\amod(G) = O\!\left( n^{\,1 - \frac{1}{d+1}} \right),
\]
and the corresponding decomposition can be computed in polynomial time when the geometric representation is given as input.
\end{theorem}

Less trivially, the tree-independence number~\cite{dallard2024treewidth} is a lower bound for this parameter. Indeed, let $S \subseteq V(G)$ be a subset of size at most $k$ such that all connected components of $G - S$ have independence number at most $k$. Construct a tree decomposition by creating a root bag $S$, and for each connected component $C$ of $G - S$, a leaf bag $C \cup S$ adjacent only to the root. Observe that the independence number of each bag is at most $2k$, which yields the following bound.

\begin{observation}\label{obs:AlphaModTreeAlpha}
For any graph $G$, $\ta(G) \leqslant 2 \amod(G)$.
\end{observation}

However, there is no direct relationship between treewidth and the $\alpha$-modulator number. Indeed, the $\alpha$-modulator number is unbounded on paths, while cliques have bounded $\alpha$-modulator number but unbounded treewidth. 

The relationships between these parameters are summarized in Figure~\ref{fig:Parameters}. In addition, we indicate in this figure the computational complexity of the \textsc{2-Subcoloring} problem. Indeed, the problem is \textsc{FPT} when parameterized by treewidth (see~\cite{fiala}), while it is \textsc{NP}-complete even on graphs of bounded tree-independence number (see~\cite{marin2025subcoloring}). 

However, in Section~\ref{sec:AlgoAlphaLeafWidth}, we show that the problem is \textsc{FPT} when parameterized by the $\alpha$-modulator number, provided that a corresponding modulator is given as part of the input. When it is not the case, this modulator can be computed in XP time $n^{O(\amod(G))}$ by simply guessing the modulator $S$ and checking that each remaining connected component of $G-S$ has independence number at most $\amod(G)$. Unfortunately, we cannot hope for a better strategy in the general case, as a simple reduction (namely, the join of a graph with itself) allows us to inherit all known (parameterized) inapproximability from \textsc{Maximum Independent Set}.

\begin{observation}
For any graph graph $G$, $\frac{1}{2}\alpha(G) \leqslant \amod(G+G) \leqslant \alpha(G)$.
\end{observation}

\begin{proof}
Denote by $G'$ the graph $G_1+G_2$, where $G_1$ and $G_2$ are two distinct copies of $G$. The inequality $\amod(G')\leqslant \alpha(G)$ is trivial by considering an empty modulator.
Then, by contradiction, suppose that $\amod(G') < \frac{1}{2}\alpha(G)$, and let $S$ be the corresponding modulator. Since $|S|<\frac{1}{2}\alpha(G)$, $S$ does not fully contain $V(G_1)$ or $V(G_2)$, and thus $G-S$ contains a unique connected component $C$. Since $\alpha(G[C])<\frac{1}{2}\alpha(G)$, we obtain $\alpha(G) \leqslant |S|+\alpha(G[C]) < \alpha(G)$, which is a contradiction. Therefore, $\amod(G') \geqslant \frac{1}{2}\alpha(G)$.
\end{proof}

In particular, assuming GAP-ETH, for any computable function $g(k) > k$, there is no algorithm running in time $f(k)\cdot n^{o(k)}$, for any computable function $f$, that, given an $n$-vertex graph $G$ and an integer $k$, distinguishes between the cases $\amod(G) \leq k$ and $\amod(G) \geq g(k)$~\cite{chalermsook2020Inapprox}. Note that the same result holds for the tree independence number~\cite{dallard2025computing}.

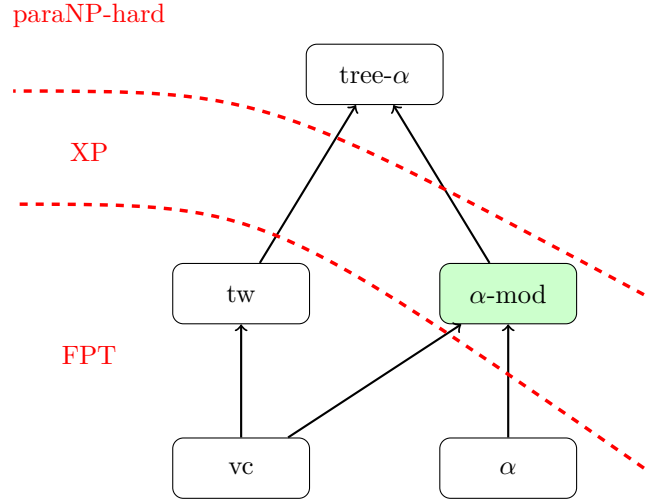
\begin{figure}[ht]
\centering
\begin{tikzpicture}[
    node distance=1.5cm and 1.7cm,
    every node/.style={draw, rounded corners, minimum width=1.8cm, minimum height=0.8cm, align=center},
    new/.style={fill=green!20},
    arrow/.style={->, thick}
]

\node (vc) {$\mathrm{vc}$};
\node (alpha) [right=of vc] {$\alpha$};

\node (tw) [above=of vc] {$\tw$};
\node (amod) [above=of alpha, new] {$\amod$};
\node (ta) [above= 2.5cm of $(tw)!0.5!(amod)$] {$\ta$};

\draw[arrow] (vc) -- (tw);
\draw[arrow] (tw) -- (ta);

\draw[arrow] (amod) -- (ta);

\draw[arrow] (vc) -- (amod);
\draw[arrow] (alpha) -- (amod);


\draw[dashed, red, line width = 1.3  ] (5.3,0) .. controls (0.3,3.5) .. (-3,3.5);

\draw[dashed, red, line width = 1.3  ] (5.3,2.3) .. controls (0.3,5) .. (-3,5);

\node[draw = none] () at (-2,1.5) {\color{red}  FPT};
\node[draw=none, align=center] at (-2,4.2) {\color{red} XP};
\node[draw = none] () at (-2,6) {\color{red}  paraNP-hard};

\end{tikzpicture}
\caption{Hierarchy of graph parameters. An arrow from parameter $p_1$ to parameter $p_2$ implies that if $p_1$ is bounded, then $p_2$ is also bounded. The parameter in green is newly introduced. The parameterized complexity landscape of the \textsc{2-Subcoloring} problem is depicted by the red dots.}
\label{fig:Parameters}
\end{figure}

\paragraph*{Link with $\mathcal{G} \rhd \mathcal{H}$ modulators.}
In~\cite{fomin2026planarity}, Fomin \emph{et~al.} introduced the notion of $\mathcal{G} \rhd \mathcal{H}$ in order to unify several width parameters based on the concept of a \emph{modulator}, such as $\mathcal{H}$-modulator number, $\mathcal{H}$-treewidth, and related notions.

Let $\mathcal{G}$ and $\mathcal{H}$ be two graph classes. The class $\mathcal{G} \rhd \mathcal{H}$ consists of all graphs $G$ that contain a subset $X \subseteq V(G)$ (called a \emph{$\mathcal{G} \rhd \mathcal{H}$-modulator}) such that $\torso(G, X) \in \mathcal{G}$ and, for every connected component $C$ of $G - X$, we have $C \in \mathcal{H}$. Here, $\torso(G, X)$ is the subgraph induced by $X$ where we also add edges between pairs of vertices with neighbors in a same connected component of $G - X$. In particular, for each connected component $C$ of $G-X$, the neighborhood $N_G(C)$ of $C$ induces a clique in $\torso(G,X)$.

For an integer $k \in \mathbb{N}$, let $\mathcal{S}_k$ denote the class of graphs with at most $k$ vertices and $\mathcal{I}_k$ the class of graphs with independence number at most $k$. For instance, $\mathcal{S}_k \rhd \mathcal{S}_{1}$ is precisely the class of graphs with a vertex cover of size at most $k$. We refer to~\cite{fomin2026planarity} for more details.

\smallskip

The $\alpha$-modulator number naturally fits into this framework (see Figure~\ref{fig:Torso}). More precisely, for a graph $G$, the minimum integer $k$ such that $G \in \mathcal{S}_k \rhd \mathcal{I}_k$ is exactly the $\alpha$-modulator number.

\begin{observation}\label{lemma:EquivalenceTorso}
For every graph $G$, $\amod(G) = \min \{k \mid G \in \mathcal{S}_k \rhd \mathcal{I}_k\}$.
\end{observation}

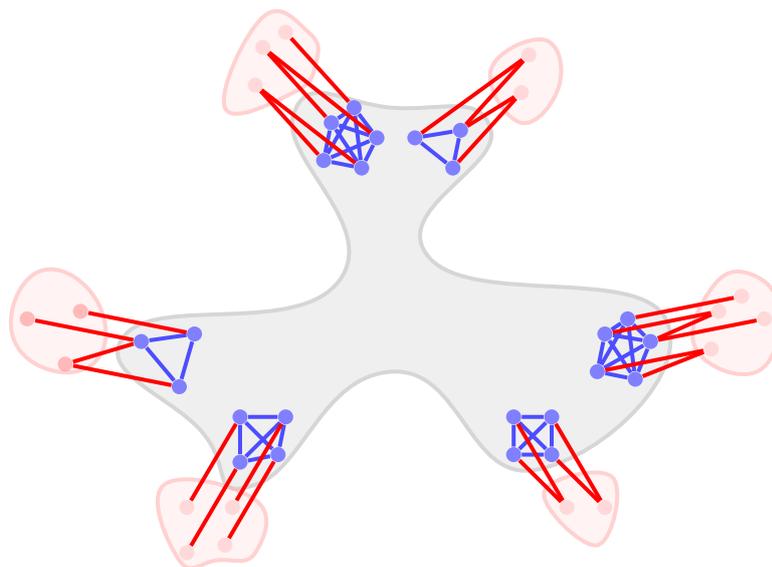
\begin{figure}
    \centering
    
\begin{tikzpicture}[
  node/.style={draw=none, inner sep=3pt},
  clique/.style={circle, fill=blue!50, inner sep=2pt},
  edge/.style={line width=2.5pt, rounded corners=8pt},
  clique_edge/.style={line width=1.5pt, color=blue!70},
  external/.style={circle, fill=red!50, inner sep=2pt, opacity=0.5},
  patate/.style={fill=red!8, draw=red!30, line width=1.5pt, opacity=0.6}
]


\node[node] (A) at (-2,-1) {};
\node[node] (B) at (0,1.5) {};
\node[node] (C) at (2,-1) {};

\node[node] (A1) at (-3,-0.8) {};
\node[node] (A2) at (-1.8,-1.8) {};
\node[node] (B1) at (-0.6,2) {};
\node[node] (B2) at (0.6, 2) {};
\node[node] (C1) at (1.8,-1.8) {};
\node[node] (C2) at (3,-0.8) {};

\begin{scope}
  \fill[black!8,draw=black!20,line width=1.5pt, opacity = 0.8]
    plot [smooth cycle, tension=0.9]
    coordinates {
      (0,-1)  (-1.8,-2.5) (-2.5,-1.8) (-3.5,-0.5)
      (-0.7, 0.2) (-1.3,2.4) (0,2.5) (1.3,2.2) (0.5,0.5)
      (3.6,-0.3)  (1.9,-2.3)
    };
\end{scope}

\node[clique] (A1c1) at ($ (A1) + (0.4,0.3) $) {};
\node[clique] (A1c2) at ($ (A1) + (0.2,-0.4) $) {};
\node[clique] (A1c3) at ($ (A1) + (-0.3,0.2) $) {};
\draw[clique_edge] (A1c1) -- (A1c2) -- (A1c3) -- (A1c1);
\begin{scope}
  \fill[patate]
    plot [smooth cycle, tension=0.8]
    coordinates {
      (-5,-0.5)  (-4.7,0.3) (-4,0.1) (-3.8,-0.7) (-4.4,-1)
    };
\end{scope}
\node[external] (A1e1) at ($ (-4.5,-0.5) + (0.4,0.3) $) {};
\node[external] (A1e2) at ($ (-4.5,-0.5) + (0.2,-0.4) $) {};
\node[external] (A1e3) at ($ (-4.5,-0.5) + (-0.3,0.2) $) {};

\draw[clique_edge, color=red] (A1c1) -- (A1e1);
\draw[clique_edge, color=red] (A1c3) -- (A1e2);
\draw[clique_edge, color=red] (A1c2) -- (A1e2);
\draw[clique_edge, color=red] (A1c3) -- (A1e3);

\node[clique] (A2c1) at ($ (A2) + (0.4,0.2) $) {};
\node[clique] (A2c2) at ($ (A2) + (0.3,-0.3) $) {};
\node[clique] (A2c3) at ($ (A2) + (-0.2,0.2) $) {};
\node[clique] (A2c4) at ($ (A2) + (-0.2,-0.4) $) {};
\draw[clique_edge] (A2c1) -- (A2c2) -- (A2c3) -- (A2c4) -- (A2c1) -- (A2c3);
\draw[clique_edge] (A2c4) -- (A2c2) ;

\node[external] (A2e1) at ($ (-2.5,-3) + (0.4,0.2) $) {};
\node[external] (A2e2) at ($ (-2.5,-3) + (0.3,-0.3) $) {};
\node[external] (A2e3) at ($ (-2.5,-3) + (-0.2,0.2) $) {};
\node[external] (A2e4) at ($ (-2.5,-3) + (-0.2,-0.4) $) {};

\begin{scope}
  \fill[patate]
    plot [smooth cycle, tension=0.8]
    coordinates {
      (-2.8,-3.5) (-3,-2.5)  (-1.8,-2.7) (-1.8,-3.4) (-2.3,-3.5)
    };
\end{scope}
\draw[clique_edge, color=red] (A2c1) -- (A2e1);
\draw[clique_edge, color=red] (A2c2) -- (A2e2);
\draw[clique_edge, color=red] (A2c3) -- (A2e3);
\draw[clique_edge, color=red] (A2c4) -- (A2e4);

\node[clique] (B1c1) at ($ (B1) + (0.4,0.1) $) {};
\node[clique] (B1c2) at ($ (B1) + (0.2,-0.3) $) {};
\node[clique] (B1c3) at ($ (B1) + (-0.2,0.3) $) {};
\node[clique] (B1c4) at ($ (B1) + (-0.3,-0.2) $) {};
\node[clique] (B1c5) at ($ (B1) + (0.1,0.5) $) {};
\draw[clique_edge] (B1c1) -- (B1c2) -- (B1c3) -- (B1c4) -- (B1c5) -- (B1c1) -- (B1c3);
\draw[clique_edge] (B1c4) -- (B1c2) -- (B1c5) ;
\draw[clique_edge] (B1c4) -- (B1c1);
\draw[clique_edge] (B1c3) -- (B1c5);

\node[external] (B1e3) at ($ (-1.5,3) + (-0.2,0.3) $) {};
\node[external] (B1e4) at ($ (-1.5,3) + (-0.3,-0.2) $) {};
\node[external] (B1e5) at ($ (-1.5,3) + (0.1,0.5) $) {};

\begin{scope}
  \fill[patate]
    plot [smooth cycle, tension=0.8]
    coordinates {
      (-2.2,2.5) (-1.7,3.7) (-1.1,3.5) (-1,3.1) (-1.3,2.7)
    };
\end{scope}
\draw[clique_edge, color=red] (B1c1) -- (B1e3);
\draw[clique_edge, color=red] (B1c2) -- (B1e4);
\draw[clique_edge, color=red] (B1c3) -- (B1e3);
\draw[clique_edge, color=red] (B1c4) -- (B1e4);
\draw[clique_edge, color=red] (B1c5) -- (B1e5);

\node[clique] (B2c1) at ($ (B2) + (0.3,0.2) $) {};
\node[clique] (B2c2) at ($ (B2) + (0.2,-0.3) $) {};
\node[clique] (B2c3) at ($ (B2) + (-0.3,0.1) $) {};
\draw[clique_edge] (B2c1) -- (B2c2) -- (B2c3) -- (B2c1);

\node[external] (B2e1) at ($ (1.5,3) + (0.3,0.2) $) {};
\node[external] (B2e2) at ($ (1.5,3) + (0.2,-0.3) $) {};

\begin{scope}
  \fill[patate]
    plot [smooth cycle, tension=0.8]
    coordinates {
      (1.3,3)  (2,3.4) (2.2,2.8) (1.7,2.3)
    };
\end{scope}
\draw[clique_edge, color=red] (B2c1) -- (B2e1);
\draw[clique_edge, color=red] (B2c2) -- (B2e2);
\draw[clique_edge, color=red] (B2c3) -- (B2e1);
\draw[clique_edge, color=red] (B2c1) -- (B2e2);

\node[clique] (C1c1) at ($ (C1) + (0.3,0.2) $) {};
\node[clique] (C1c2) at ($ (C1) + (0.3,-0.3) $) {};
\node[clique] (C1c3) at ($ (C1) + (-0.2,0.2) $) {};
\node[clique] (C1c4) at ($ (C1) + (-0.2,-0.3) $) {};
\draw[clique_edge] (C1c1) -- (C1c2) -- (C1c3) -- (C1c4) -- (C1c1) -- (C1c3);
\draw[clique_edge] (C1c4) -- (C1c2) ;
\node[external] (C1e1) at ($ (2.5,-3) + (0.3,0.2) $) {};
\node[external] (C1e3) at ($ (2.5,-3) + (-0.2,0.2) $) {};

\begin{scope}
  \fill[patate]
    plot [smooth cycle, tension=0.8]
    coordinates {
      (2.7,-3.3) (2.9,-2.4) (2,-2.6)
    };
\end{scope}
\draw[clique_edge, color=red] (C1c1) -- (C1e1);
\draw[clique_edge, color=red] (C1c2) -- (C1e1);
\draw[clique_edge, color=red] (C1c3) -- (C1e3);
\draw[clique_edge, color=red] (C1c4) -- (C1e3);

\node[clique] (C2c1) at ($ (C2) + (0.4,0.2) $) {};
\node[clique] (C2c2) at ($ (C2) + (0.2,-0.3) $) {};
\node[clique] (C2c3) at ($ (C2) + (-0.2,0.3) $) {};
\node[clique] (C2c4) at ($ (C2) + (-0.3,-0.2) $) {};
\node[clique] (C2c5) at ($ (C2) + (0.1,0.5) $) {};
\draw[clique_edge] (C2c1) -- (C2c2) -- (C2c3) -- (C2c4) -- (C2c5) -- (C2c1) -- (C2c3);
\draw[clique_edge] (C2c4) -- (C2c2) -- (C2c5) ;
\draw[clique_edge] (C2c4) -- (C2c1);
\draw[clique_edge] (C2c3) -- (C2c5);

\node[external] (C2e1) at ($ (4.5,-0.5) + (0.4,0.2) $) {};
\node[external] (C2e3) at ($ (4.5,-0.5) + (-0.2,0.3) $) {};
\node[external] (C2e4) at ($ (4.5,-0.5) + (-0.3,-0.2) $) {};
\node[external] (C2e5) at ($ (4.5,-0.5) + (0.1,0.5) $) {};

\begin{scope}
  \fill[patate]
    plot [smooth cycle, tension=0.8]
    coordinates {
      (5.1,-0.2) (4.7,0.3) (4.2,0.1) (4,-0.7) (4.8,-1)
    };
\end{scope}
\draw[clique_edge, color=red] (C2c1) -- (C2e1);
\draw[clique_edge, color=red] (C2c1) -- (C2e3);
\draw[clique_edge, color=red] (C2c2) -- (C2e4);
\draw[clique_edge, color=red] (C2c3) -- (C2e3);
\draw[clique_edge, color=red] (C2c4) -- (C2e4);
\draw[clique_edge, color=red] (C2c5) -- (C2e5);

\end{tikzpicture}
    \caption{An illustration of the parameter $\min \{k\mid G\in \mathcal{S}_k \rhd \mathcal{I}_k\}$, which is equal to the $\alpha$-modulator number. The middle part (in gray) represents the torso $\text{\sf torso}(G,X)$, which has size at most $k$. The red parts correspond to the connected components of $G - X$, each having independence number at most $k$.}
\label{fig:Torso}
\end{figure}

\section{Example of applications}\label{sec:AlgoAlphaLeafWidth}

In Section~\ref{sec:SepWidth}, we introduced a parameter, the $\alpha$-modulator number, which is sublinear in intersection graphs of similarly sized $\beta$-fat objects in $\mathbb{R}^d$. 

In this section, we prove that the \textsc{2-Subcoloring} problem is \textsc{FPT} when parameterized by the $\alpha$-modulator number, provided that a corresponding modulator is given as input. We also show that the \textsc{Two Sets Cut-Uncut} problem is in \textsc{XP} with respect to this parameter.

\smallskip
To design an FPT/XP algorithm parameterized by $\amod(G)$, two conditions are necessary. First, the problem must be FPT/XP parameterized by the independence number, in order to handle the connected components after removing the modulator. Second, it must be FPT/XP parameterized by the vertex cover number, to support the dynamic programming approach on the modulator. The main difficulty lies in combining these two algorithmic components coherently. 

\subsection{2-Subcoloring}

Recall that a $2$-subcoloring of a graph $G$ is partition of the vertex set $V(G)$ into two sets $A$ and $B$ such that both $G[A]$ and $G[B]$ are cluster graphs. An equivalent formulation, which will be useful in the following, is that a $2$-subcoloring of a graph is a $2$-coloring $\mu : V(G)\rightarrow \{0,1\}$ of $V(G)$ with no monochromatic induced path on three vertices. 

In~\cite{fiala}, it was shown that \textsc{2-Subcoloring} can be solved in time $2^{O(\tw(G))}\cdot n$ using a standard dynamic programming algorithm over a nice tree decomposition, with one additional technical feature: for each partial coloring of a bag, the algorithm associates a Boolean flag to each clique in the bag, indicating whether some vertices of that clique have been forgotten in the decomposition.

In~\cite{kanj2018parameterized}, Kanj \emph{et~al.}\ proved that \textsc{2-Subcoloring} is FPT when parameterized by the total number of clusters in the resulting subcoloring. It is easy to observe that this number is at most twice the independence number $\alpha(G)$ of the input graph, which directly yields an FPT algorithm parameterized by $\alpha(G)$. They were even able to deal with a partial version of the problem.

\begin{theorem}[Theorem~1.4~\cite{kanj2018parameterized}]\label{thm:FPTSubcoloring}
Let $G$ be a graph, $c$ an integer, and $(A',B')$ a partial $2$-subcoloring of $G$.
It is possible to decide whether there exists a $2$-subcoloring $(A,B)$ of $G$ such that $A'\subseteq A$, $B'\subseteq B$, and the total number of clusters in $(A,B)$ is at most $c$, in time $2^{O(c)}\cdot n^2$.
\end{theorem}

\medskip

We are now equipped to present the main algorithm, which borrows some ideas of the dynamic programming on tree decompositions presented in~\cite{fiala}.

\begin{theorem}
The \textsc{$2$-Subcoloring} problem can be solved in time $2^{O(\amod(G))}\cdot n^{O(1)}$, when the corresponding modulator is given in input.
\end{theorem}

\begin{proof}

We start by introducing definitions that will be useful for our dynamic programming algorithm. Let $G$ be a graph and let $\mu : V(G) \rightarrow \{0,1\}$ be a $2$-subcoloring of $G$.  
For a vertex $v\in V(G)$, we denote by 
\[
C_{\mu}(v) = \{u\in N_G[v] \mid \mu(u) = \mu(v)\}
\]
the clique containing $v$ in the $2$-subcoloring $\mu$.  

Given a subset of vertices $X\subseteq V(G)$ and a function $\sigma : X\rightarrow \{\bot,\top\}$, we say that $\sigma$ is a \emph{signature of $\mu$ on $X$} if:
\begin{enumerate}[(i)]
\item for every $v\in X$ and every $u\in C_\mu(v)\cap X$, we have $\sigma(u)=\sigma(v)$; 
\item for every $v\in X$ with $\sigma(v)=\bot$, the whole clique $C_\mu(v)$ is contained in $X$.
\end{enumerate}

Intuitively, a signature $\sigma$ records how the clusters defined by $\mu$ interact with $X$: vertices in the same cluster receive the same value, and $\bot$ indicates that the cluster of the vertex is entirely contained in $X$, i.e., there is no "hidden" vertex outside $X$ in the cluster.  

If, in addition, for every $v\in X$, the implication $C_\mu(v)\subseteq X$ implies $\sigma(v)=\bot$, then $\sigma$ is called the \emph{exact signature} of $\mu$ on $X$. While a subcoloring may have several signatures, it has a unique exact signature.

\medskip
We now describe the dynamic programming algorithm. Let $S\subseteq V(G)$ be such that $|S|\leqslant k$, and for every connected component $C$ of $G-S$, we have $\alpha(G[C])\leqslant k$. Let $C_1,\dots,C_q$ ($q\geq 0$) be the connected components of $G-S$. For $0\leqslant t\leqslant q$, define
\[
V_t = S \cup \bigcup_{i=1}^t C_i.
\]
We maintain a table $\mathbf{Tab}_t$ whose entries are pairs
\[
(\mu,\sigma), \qquad \mu : S \to \{0,1\}, \ \sigma : S \to \{\top,\bot\},
\]
with the following meaning: $(\mu,\sigma)\in \mathbf{Tab}_t$ if and only if $G[V_t]$ admits a $2$-subcoloring $\nu$ such that $\nu\restriction S=\mu$ and $\sigma$ is a signature of $\nu$ on $S$.

Observe that $G$ admits a $2$-subcoloring if and only if $\mathbf{Tab}_q \neq \emptyset$.

\medskip
For $t=0$, we add $(\mu,\sigma)$ to $\mathbf{Tab}_0$ if and only if $\mu$ is a $2$-subcoloring of $G[S]$ and $\sigma$ verifies the point (i) of the definition of a signature. In this case, $\sigma$ is trivially a signature of $\mu$ on $S$.

\medskip
Assume that $\mathbf{Tab}_{t-1}$ has been computed for some $t\geqslant 1$. Let $(\mu,\sigma_1)\in \mathbf{Tab}_{t-1}$, and let $\sigma_2 : S \to \{\top,\bot\}$ such that
\[
\sigma_1(v)=\top \ \Rightarrow\ \sigma_2(v)=\top \quad \text{for all } v\in S.
\]
We test whether $(\mu,\sigma_2)$ can be added to $\mathbf{Tab}_t$.

First, we discard $(\mu,\sigma_2)$ if there exist adjacent vertices $u,v\in S$ with $\sigma_2(u)\neq \sigma_2(v)$ and $\mu(u)=\mu(v)$, as this would violate condition (i) of a signature.  

Next, we discard the pair if there exists a vertex $u\in C_t$ having two neighbors $v,v'\in S$ such that $\sigma_2(v)=\sigma_2(v')=\bot$ but $\mu(v)\neq \mu(v')$, as any color assignment to $u$ would violate condition (ii) of a signature.

Otherwise, we define a partial coloring $\tilde{\mu}$ on $S\cup C_t$ as follows: for every $v\in C_t$, if there exists a neighbor $u\in S$ such that $\sigma_2(u)=\bot$ or $\sigma_1(u)=\top$, then we set
\[
\tilde{\mu}(v)=1-\mu(u).
\]
If this assignment is inconsistent or does not define a valid partial $2$-subcoloring, we discard the pair.

Otherwise, we run the algorithm of Theorem~\ref{thm:FPTSubcoloring} on $G[S\cup C_t]$ with the partial coloring $\tilde{\mu}$. Since $\alpha(G[S\cup C_t])\leq 2k$, this can be done in time $2^{O(k)}\cdot n^{O(1)}$. If a valid extension is found, we add $(\mu,\sigma_2)$ to $\mathbf{Tab}_t$. We call the pair $(\mu, \sigma_1)\in \mathbf{Tab}_{t-1}$ the \emph{ancestor} of $(\mu, \sigma_2)\in \mathbf{Tab}_t$

\medskip
We now prove that the transitions are correct.

\smallskip
\emph{Completeness.}
Assume that $G[V_t]$ admits a $2$-subcoloring $\nu$ with signature $\sigma_2$ on $S$. Let $\sigma_1$ be the exact signature of $\nu\restriction V_{t-1}$ on $S$. Then $(\mu,\sigma_1)\in \mathbf{Tab}_{t-1}$ by induction.

Moreover, for every $v\in S$, $\sigma_2(v)=\bot$ implies $\sigma_1(v)=\bot$. By construction of $\tilde{\mu}$, every vertex of $C_t$ that is constrained by $S$ receives the same color as in $\nu$, hence $\tilde{\mu}$ is consistent with $\nu$. Therefore, the algorithm of Theorem~\ref{thm:FPTSubcoloring} finds a valid extension, and $(\mu,\sigma_2)$ is added to $\mathbf{Tab}_t$.

\smallskip
\emph{Soundness.}
Suppose that $(\mu,\sigma_2)$ is added to $\mathbf{Tab}_t$. Then there exists an ancestor pair $(\mu,\sigma_1)\in \mathbf{Tab}_{t-1}$ and a $2$-subcoloring $\tilde{\nu}$ of $G[S\cup C_t]$ extending $\tilde{\mu}$. Let $\nu$ be a $2$-subcoloring of $G[V_{t-1}]$ such that $\nu\restriction S = \mu$ and which has $\sigma_1$ as a signature on $S$, which exists by induction hypothesis.

We define $\lambda$ on $V_t$ by combining $\nu$ and $\tilde{\nu}$. This is well-defined since both agree on $S$.

We claim that $\lambda$ is a $2$-subcoloring. Otherwise, there exists a monochromatic path $a b c$. Since $\lambda \restriction V_{t-1}$ and $\lambda \restriction (S\cup C_t)$ are both $2$-subcolorings, such a path must intersect both $V_{t-1}\setminus S$ and $C_t$, hence $b\in S$. Note that $b$ has a neighbor with the same color in $V_{t-1}$. By correction of $\mathbf{Tab}_{t-1}$, it follows that $\sigma_1(b)=\top$ and thus the construction of $\tilde{\mu}$ forces $\tilde{\mu}(c)=1-\mu(b)$, a contradiction.

Finally, one checks that $\sigma_2$ satisfies the definition of a signature for $\lambda$ on $S$ since otherwise the pair would have been rejected directly.

\medskip
Each table $\mathbf{Tab}_t$ contains at most $2^{2k}$ entries, and each of them can be computed in time $2^{O(k)}\cdot n^{O(1)}$. The overall complexity is $2^{O(k)}\cdot n^{O(1)}$, which completes the proof.
\end{proof}

As a corollary from Theorem~\ref{thm:BoundAlphaMod} and the result above, we obtain directly a subexponential algorithm for intersection graphs of similarly-sized fat objects.

\begin{theorem}\label{thm:Subexp2Subcoloring}
For every constants $d\geqslant 2$ and $\beta \geqslant 1$, there is an algorithm solving \textsc{2-Subcoloring} in time $2^{O(n^{1-1/(d+1)})}$ on intersection graphs of similarly sized $\beta$-fat objects in $\mathbb{R}^d$.
\end{theorem}

\subsection{Two Sets Cut-Uncut}

Recall that the \textsc{Two Sets Cut-Uncut} problem asks, given an edge-weighted graph $(G,w)$ and two sets of distinct terminal sets $S$ and $T$, to find an $S$-$T$-cut $(A,B)$ of minimum weight such that $S$ (resp.\ $T$) is connected in $G[A]$ (resp.\ $G[B]$). In~\cite{bentert2024parameterized} it was proved that this problem can be solved in time $n^{O(\alpha(G)^2)}$ on $n$-vertex graph $G$. We start by observing that it is possible to drop the quadratic dependency when the input graph is the intersection graph of similarly sized fat objects.
\smallskip 
We begin with two lemmas previously proved in~\cite{de2018framework}.  
Given a constant $\kappa \geqslant 1$ and a graph $G$, a \emph{$\kappa$-partition of $G$} is a partition $\mathcal{P} = (V_1,\ldots,V_q)$ of the vertex set $V(G)$ such that each $V_i$ can be partitioned into at most $\kappa$ cliques.  
In particular, when $\kappa = 1$, each part $V_i$ must itself be a clique in $G$.  
A partition $\mathcal{P} = (V_1,\ldots,V_q)$ is said to be \emph{greedy} if there exists a maximal independent set $S \subseteq V(G)$ such that each $V_i$ contains exactly one vertex of $S$.  
The \emph{$\mathcal{P}$-contraction} of $G$, denoted $G_{\mathcal{P}}$, is the graph whose vertex set is $\mathcal{P}$, where two parts $V_i$ and $V_j$ (with $i \neq j$) are adjacent if and only if there exists an edge in $G$ between a vertex of $V_i$ and a vertex of $V_j$.

The first lemma shows that one can always obtain a greedy $\kappa$-partition $\mathcal{P}$ for which the contracted graph $G_{\mathcal{P}}$ has bounded maximum degree.

\begin{lemma}[Lemma~9, \cite{de2018framework}]\label{lemma:KappaPartition}
Let $d \geqslant 2$ and $\beta \geqslant 1$ be constants.  
Then there exist constants $\kappa$ and $\Delta$ such that for any intersection graph $G$ of $n$ similarly sized $\beta$-fat objects in $\mathbb{R}^d$, a greedy $\kappa$-partition $\mathcal{P}$ for which $G_{\mathcal{P}}$ has maximum degree at most $\Delta$ can be computed in polynomial time.
\end{lemma}

Given a set of terminals $S \subseteq V(G)$, a \emph{Steiner tree} for $S$ is a vertex set $X$ such that $S \subseteq X$ and $G[X]$ is connected.  

\begin{lemma}[Lemma~18, \cite{de2018framework}]\label{lemma:SteinerTree}
Let $\mathcal{P}$ be a $\kappa$-partition of a graph $G$ such that $G_{\mathcal{P}}$ has maximum degree $\Delta$.  
Suppose $X$ is a minimal Steiner tree for a set of terminals $S$ (that is, no proper subset $X' \subset X$ is also a Steiner tree for $S$).  
Then $X$ contains at most $\kappa^2 (\Delta + 1)$ vertices from each partition class that do not belong to $S$.
\end{lemma}

We can now prove the main result, namely an XP algorithm for the \textsc{Two Sets Cut-Uncut} problem parameterized by the $\alpha$-leaf width of the input graph. The approach is similar to the one used for the \textsc{2-Subcoloring} problem: we solve the problem on connected components after the removal of the modulator, and then recombine the solutions. 

\begin{theorem}\label{thm:XPCutUncut}
Let $d\geqslant 2$ and $\beta \geqslant 1$ be two constants. There exists an algorithm which, given an intersection graph $G$ of $n$ similarly sized $\beta$-fat objects in $\mathbb{R}^d$, solves the \textsc{Two Sets Cut-Uncut} problem on $G$ in time $n^{O(\amod(G))}$.
\end{theorem}
\begin{proof}
Let $(G,w)$ be an edge-weighted graph and let $S,T\subseteq V(G)$ be two disjoint terminal sets. Suppose we are given a set $X\subseteq V(G)$ of size at most $k$ such that every connected component $C$ of $G-X$ satisfies $\alpha(G[C])\leqslant k$. Let $C_1,\dots,C_q$ ($q\ge 0$) be the connected components of $G-X$. For $0\leqslant t\leqslant q$, define
\[
V_t = X \cup \bigcup_{i=1}^t C_i
\quad \text{and} \quad
G_t = G[V_t].
\]

For every $0\leqslant t \leqslant q$, we define a table
\[
\mathbf{Tab}_t[\sigma,\Pi_A,\Pi_B,\mu_A,\mu_B]
\in \mathbb{N}\cup\{+\infty\},
\]
where:
\begin{itemize}
    \item $\sigma : X \to \{A,B\}$ assigns each vertex of $X$ to a side, respecting $S \cap X \subseteq \sigma^{-1}(A)$ and $S \cap B \subseteq \sigma^{-1}(B)$;
    \item $\Pi_A$ is a partition of $\sigma^{-1}(A)$ and $\Pi_B$ is a partition of $\sigma^{-1}(B)$;
    \item $\mu_A \subseteq \Pi_A$ and $\mu_B \subseteq \Pi_B$, which will serve to identify the connected components which already contain terminals.
\end{itemize}

The value $\mathbf{Tab}_t[\sigma,\Pi_A,\Pi_B,\mu_A,\mu_B]$ stores the minimum weight of a solution cut $(A,B)$ of $G_t$ consistent with:
\begin{enumerate}
    \item the assignment $\sigma$ on $X$;
    \item each part of $\Pi_A$ (resp.\ $\Pi_B$) is exactly the intersection of a connected component of $G_t[A]$ (resp.\ $G_t[B]$) with $X$;
    \item if a part is not in $\mu_A$ (resp.\  $\mu_B$), then its corresponding connected component in $G[A]$ (resp.\ $G[B]$) does not contain any vertex from $S$ (resp.\ $T$).
\end{enumerate}

If such a partial solution exists, the state is called \emph{feasible in $G_t$}; otherwise its value is $+\infty$ and it is \emph{infeasible}.

\medskip
Since $|X|\leqslant k$, the number of possible states is at most $2^k \cdot k^{2k} = k^{O(k)}$.

\medskip
\noindent\textbf{Computation of $\mathbf{Tab}_0$.}
For each assignment $\sigma : X \to \{A,B\}$, let $\Pi_A$ (resp.\ $\Pi_B$) be the partition of $\sigma^{-1}(A)$ (resp.\ $\sigma^{-1}(B)$) into connected components of $G_0[\sigma^{-1}(A)]$ (resp.\ $G_0[\sigma^{-1}(B)]$). 

Let $\mu_A \subseteq \Pi_A$ (resp.\ $\mu_B \subseteq \Pi_B$) be the set of parts that contain at least one vertex of $S \cap X$ (resp.\ $T \cap X$).

If $\sigma(v)=A$ for every $v \in S \cap X$ and $\sigma(v)=B$ for every $v \in T \cap X$, then we set
\[
\mathbf{Tab}_0[\sigma,\Pi_A,\Pi_B,\mu_A,\mu_B]=w\big(\delta(\sigma^{-1}(A),\sigma^{-1}(B))\big),
\]
where the cut is taken in $G_0 = G[X]$. Otherwise, we set this value to $+\infty$.

This correctly initializes the table, since $G_0$ contains only the vertices of $X$.

\medskip
\noindent\textbf{Computation of $\mathbf{Tab}_t$.}
Fix $t \geqslant 1$ and assume that $\mathbf{Tab}_{t-1}$ is correct and complete. 
We describe how to compute $\mathbf{Tab}_t$ from $\mathbf{Tab}_{t-1}$.

Let $(\sigma, \Pi_A, \Pi_B, \mu_A, \mu_B)$ be a state such that 
$\mathbf{Tab}_{t-1}[\sigma, \Pi_A, \Pi_B, \mu_A, \mu_B] < +\infty$. 
We extend this partial solution by incorporating the component $C_t$.

\medskip
\noindent
\emph{Guessing the interaction with $X$.}
We enumerate all pairs of partitions $(\Pi_A', \Pi_B')$ obtained from $(\Pi_A, \Pi_B)$ by merging parts. 
Intuitively, $\Pi_A'$ (resp.\ $\Pi_B'$) represents the connectivity on the $A$-side (resp.\ $B$-side) after adding vertices of $C_t$.

\medskip
\noindent
\emph{Guessing Steiner trees inside $C_t$.}
We guess families of pairwise disjoint sets
\[
A_1',\dots,A_a' \subseteq C_t
\quad \text{and} \quad
B_1',\dots,B_b' \subseteq C_t,
\]
such that each $A_i'$ and $B_j'$ is an independent set and 
\[
\sum_i |A_i'| + \sum_j |B_j'| \leqslant 2k.
\]
These sets will serve as dominating sets of the connected components that intersect $C_t$ on each side.

For each $A_i'$ (resp.\ $B_j'$), we further guess a set $X_i \subseteq C_t$ (resp.\ $Y_j \subseteq C_t$) inducing a connected subgraph of $G[C_t]$ that spans $A_i'$ (resp.\ $B_j'$). 
By Lemma~\ref{lemma:SteinerTree}, we may assume $|\bigcup_{i}X_i|,|\bigcup_{j}Y_j| = O(k)$, and thus all these sets can be guessed in time $n^{O(k)}$.

\medskip
\noindent
\emph{Constructing a constrained cut instance.}
We now determine the assignment of the remaining vertices of $C_t$ by solving a minimum cut instance.

We build an auxiliary graph $H$ on vertex set $X \cup C_t$ with edge weights inherited from $G$, and add two terminals $s$ and $t$. 
We enforce the following constraints:
\begin{itemize}
    \item every vertex $v \in X$ is assigned according to $\sigma(v)$;
    \item every vertex of $S \cap C_t$ is forced to the $A$-side (connected to $s$ with infinite capacity), and every vertex of $T \cap C_t$ to the $B$-side;
    \item for each guessed set $X_i$, all its vertices are forced to the $A$-side, and for each $Y_j$, all its vertices are forced to the $B$-side;
    \item if a vertex $v \in C_t$ is adjacent to vertices in two distinct sets $A_i'$ and $A_{i'}'$ with $i \neq i'$, then $v$ is forced to the $B$-side; symmetrically for the $B$-side;
    \item if a vertex $v \in C_t$ has no neighbor in any $A_i'$, it is forced to the $B$-side, and symmetrically for $B$.
\end{itemize}

All remaining vertices are left free. We compute a minimum $(s,t)$-cut in $H$, yielding a bipartition $(A',B')$ of $X \cup C_t$.

\medskip
\noindent
\emph{Checking consistency.}
From $(A',B')$, we compute the induced partitions $\widehat{\Pi}_A$ and $\widehat{\Pi}_B$ of $X$ according to connectivity in $G_t[A']$ and $G_t[B']$. 
We keep the solution only if:
\begin{enumerate}
\item The partition $\Pi_A'$ (resp.\ $\Pi_B'$) is exactly the coarsest partition refining both $\Pi_A$ (resp.\ $\Pi_B$) and $\widehat{\Pi}_A$ (resp.\ $\widehat{\Pi}_B$). More formally, $\Pi_A'$ (resp. $\Pi_B'$) represents the transitive closure of the equivalence relations represented by $\Pi_A$ and $\widehat{\Pi}_A$ (resp. by $\Pi_B$ and $\widehat{\Pi}_B$).
\item The updated sets $\mu_A'$ and $\mu_B'$ correctly identify the parts containing terminals, that is, every part of $\Pi_A'$ (resp.\ $\Pi_B'$) is marked if and only if it contains a vertex of $S$ (resp.\ $T$) in $G_t[A']$ (resp.\ $G_t[B']$).
\end{enumerate}

\medskip
\noindent
\emph{Updating the table.}
For every valid construction as above, we update
\[
\mathbf{Tab}_t[\sigma,\Pi_A',\Pi_B',\mu_A',\mu_B']
\]
by taking the minimum over
\[
\mathbf{Tab}_{t-1}[\sigma,\Pi_A,\Pi_B,\mu_A,\mu_B] 
\;+\;
w\big(\delta(A',B') \cap (E(C_t) \cup E(C_t,X))\big),
\]
that is, we add exactly the contribution of edges incident to $C_t$, 
since all edges inside $G_{t-1}$ have already been accounted for.
\medskip
\noindent
Since the number of states is $k^{O(k)}$ and each transition can be evaluated in time $n^{O(k)}$, the computation of $\mathbf{Tab}_t$ takes time $n^{O(k)}$.

\medskip
The final result is the minimum value over all entries of $\mathbf{Tab}_q$ with $\mu_A$ and $\mu_B$ being singletons.
\end{proof}

As a direct consequence, we obtain a subexponential algorithm in the case of geometric intersection graphs.

\begin{theorem}
For every constants $d\geqslant 2$ and $\beta \geqslant 1$, there is an algorithm solving \textsc{Two Sets Cut-Uncut} in time $2^{O\left(n^{1-1/(d+1)}\ln(n)\right)}$ on intersection graphs of similarly sized $\beta$-fat objects in $\mathbb{R}^d$.
\end{theorem}

\section{A framework for matching lower bound}

\subsection{Overview}
\label{subsec:overview}

This section establishes a matching lower bound under the Exponential Time Hypothesis (ETH) for the running time of the form $2^{o(n^{1-1/(d+1)})}$. We provide a high-level sketch of the reduction, which starts from the \textsc{Monotone Not-All-Equal-3-SAT} problem, defined as follows.

\problemdef{Monotone Not-All-Equal-3-SAT}{ A negation-free 3-CNF formula $\Phi = \bigwedge_{i=1}^m (\ell_1^j\vee\ell_2^j\vee \ell_3^j) $ with $m$ clauses over $n$ variables $x_1,...,x_n$, where each variable appears exactly four times in $\Phi$.}{ A truth assignment $\tau \{x_1,...,x_n\} \rightarrow \{\textsf{false}, \textsf{true}\}$ such that every clause contains at least one \textsf{true} and one \textsf{false} literal.}

By the Sparsification Lemma~\cite{impagliazzo2001problems} and the reduction of~\cite{darmann2020simple}, this problem cannot be solved in time $2^{o(n)}$, where $n$ is the number of variables, unless the ETH fails.

\smallskip

The reduction is constructed as follows. Consider a $d$-dimensional grid of side length $n^{1/d}$, consisting of $n$ cells, each of side length $1$. We fix a BFS tree $T$ on the corresponding $d$-dimensional grid graph, where each vertex represents a cell, and edges connect vertices if their corresponding cells are adjacent. Since the grid graph has diameter $O(n^{1/d})$, the same bound holds for $T$.

Each clause of the formula is mapped to a cell in the grid (and thus to a node in $T$), and a \emph{gadget} is placed at each node. The gadget's role is to verify whether the corresponding clause contains at least one \textsf{true} and one \textsf{false} literal, and to propagate the truth values of each literal. Each vertex in a gadget is labeled with a literal from the original formula, subject to the following constraints:
\begin{itemize}
    \item For any literal $\ell$, each clause gadget contains $O(1)$ vertices labeled with $\ell$.
    \item If a literal $\ell$ appears in a clause gadget corresponding to a node $t \in V(T)$, then $\ell$ must belong to a clause assigned to a descendant of $t$ in $T$.
\end{itemize}

This property ensures that vertices labeled with a given literal $\ell$ appear only along a single branch of $T$. Since there are at most $4n$ literals, the total number of vertices in the graph is bounded by $O(n \cdot n^{1/d}) = O(n^{(d+1)/d})$.
\smallskip

At the end of the construction, the resulting graph is a \textbf{yes-instance} of the target problem if and only if the original formula is a \textbf{yes-instance} of \textsc{Monotone Not-All-Equal-3-SAT}. This construction can be realized using unit balls in $\mathbb{R}^d$ for any $d \geqslant 3$, and with similarly sized fat objects in $\mathbb{R}^2$. 
Whether the lower bound also holds for unit disk graphs remains an open question. 

\subsection{Lower bound for 2-Subcoloring}\label{sec:LowerBoundSubcoloring}

Applying the framework described above, we prove that the running time of Theorem~\ref{thm:Subexp2Subcoloring} is almost tight under the ETH, up to the logarithmic factor in the exponent.

\begin{theorem}\label{thm:LowerBound2Subcoloring}
 \textsc{2-Subcoloring} cannot be solved in time $2^{o\left(n^{1-1/(d+1)}\right)}$ on $n$-vertex intersection graphs of similarly sized fat objects in $\mathbb{R}^d$ for any $d\geqslant 2$, and even in unit ball graphs in $\mathbb{R}^d$ for any $d\geqslant 3$.
\end{theorem}

\begin{proof}
Let $\phi$ be a monotone $3$-SAT formula with $n$ variables and $m = 4n$ clauses, where each clause contains exactly three positive literals and each variable appears in exactly four clauses. We call $\mathcal{L}$ the set of literals.

Let $p$ be the smallest odd integer greater than $n^{1/d}$.  
Let $G_{d,p}$ denote the $d$-dimensional grid graph on $[p]^d$, and let $o$ be its unique central node. Given two nodes $t=(t_1,...,t_d)\in [p]^d$ and $t'=(t_1',...,t_d')\in [p]^d$ such that $tt'\in E(G_{d,p})$, call the \emph{dimension} of the edge $tt'$ the unique $d'\in [d]$ such that $|t_{d'}-t'_{d'}|=1$. We construct a BFS tree $T_{d,p}$ of $G_{d,p}$ rooted at $o$.  
Note that $T_p$ has height at most $p$, since every vertex of $G_{d,p}$ is at distance at most $dp/2$ from $o$.
\smallskip
We first describe the \emph{cell gadget} $H_{L,d}$, where $L\subseteq \mathcal{L}$ is a set of literals. Begin with a path $v_{1,1}, v_{2,1}, \dots, v_{2d+1,1}$. For each $i$ with $2 \leqslant i \leqslant d$, add four vertices $v_{1,i}, v_{2,i}, v_{3,i}, v_{4,i}$ and the edges $v_{1,i}v_{2,i}$,  $v_{2,i}v_{2i-2,1}$, $v_{2i-2,1}v_{3,i}$ and  $v_{3,i}v_{4,i}$. Finally, add a special vertex $v_{1,d+1}$ adjacent to $v_{2d,1}$. This forms the \emph{skeleton} of the gadget (Figure~\ref{fig:clause_gadget} illustrates the cases $d=2$ and $d=3$). Intuitively, the first long path is in charge of the propagation of the solution in the first dimension, the $(d-1)$ short paths attached to it are in charge of the propagation in the other dimensions, and the vertex $v_{1,d+1}$ represents the clause. Then, we replace each vertex $v_{i,j}$ in the skeleton with a clique
    $$
        C_{i,j} = \{ v_{i,j,\ell} \mid \ell \in L \}.
    $$
For each edge $v_{i,j}v_{i',j'}$ in the skeleton, connect $v_{i,j,\ell}$ to $v_{i',j',\ell}$ for every $\ell \in L$. The main challenge, when representing this gadget in a $d$-dimensional space (especially with balls as vertices), is that this blow-up will be performed in one of the $d$ dimensions, and thus we will have to manage carefully this blow up with one of the previously short path.

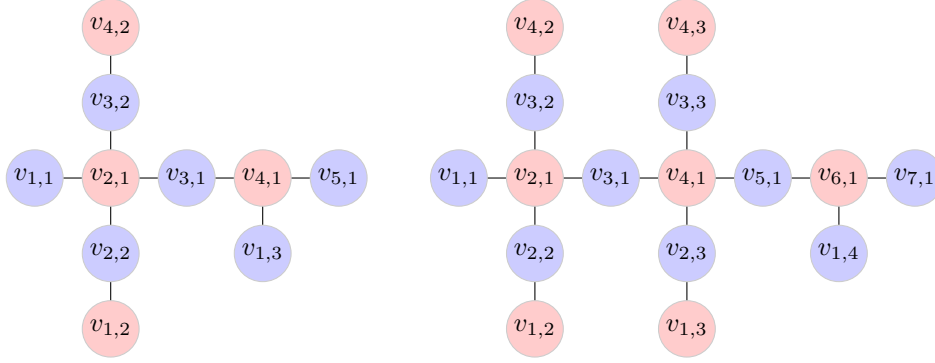
\begin{figure}[ht]
\centering
\begin{tikzpicture}[scale=1.0, every node/.style={circle, draw, inner sep=1.5pt, minimum size = 20}]
  \def\d{3} 
  \def\n{7} 

    \node[fill = blue, text opacity = 1, opacity = 0.2] (w1) at (1,0) {$v_{1,1}$};    
    \node[fill = red, text opacity = 1, opacity = 0.2] (w2) at (2,0) {$v_{2,1}$};    
    \node[fill = blue, text opacity = 1, opacity = 0.2] (w3) at (3,0) {$v_{3,1}$};    
    \node[fill = red, text opacity = 1, opacity = 0.2] (w4) at (4,0) {$v_{4,1}$};    
    \node[fill = blue, text opacity = 1, opacity = 0.2] (w5) at (5,0) {$v_{5,1}$};    
    \draw (w1) -- (w2)-- (w3)-- (w4) -- (w5) ;

  \node[fill = red, text opacity = 1, opacity = 0.2]  (v1-2) at (2,-2) {$v_{1,2}$};
  \node[fill = blue, text opacity = 1, opacity = 0.2]  (v2-2) at (2,-1) {$v_{2,2}$};
  \node[fill = blue, text opacity = 1, opacity = 0.2]  (v3-2) at (2,1) {$v_{3,2}$};
  \node[fill = red, text opacity = 1, opacity = 0.2]  (v4-2) at (2,2) {$v_{4,2}$};

  \draw (v1-2)  -- (v2-2) -- (w2) -- (v3-2) -- (v4-2);

  \node[fill = blue, text opacity = 1, opacity = 0.2]  (v1-3) at (4,-1) {$v_{1,3}$};
  \draw (v1-3)--(w4);
\end{tikzpicture}
\qquad
\begin{tikzpicture}[scale=1.0, every node/.style={circle, draw, inner sep=1.5pt, minimum size = 20}]
  \def\d{3} 
  \def\n{7} 

    \node[fill = blue, text opacity = 1, opacity = 0.2] (w1) at (1,0) {$v_{1,1}$};    
    \node[fill = red, text opacity = 1, opacity = 0.2] (w2) at (2,0) {$v_{2,1}$};    
    \node[fill = blue, text opacity = 1, opacity = 0.2] (w3) at (3,0) {$v_{3,1}$};    
    \node[fill = red, text opacity = 1, opacity = 0.2] (w4) at (4,0) {$v_{4,1}$};    
    \node[fill = blue, text opacity = 1, opacity = 0.2] (w5) at (5,0) {$v_{5,1}$};    
    \node[fill = red, text opacity = 1, opacity = 0.2] (w6) at (6,0) {$v_{6,1}$};
    \node[fill = blue, text opacity = 1, opacity = 0.2] (w7) at (7,0) {$v_{7,1}$};
    \draw (w1) -- (w2)-- (w3)-- (w4) -- (w5) -- (w6) -- (w7);

  \node[fill = red, text opacity = 1, opacity = 0.2]  (v1-2) at (2,-2) {$v_{1,2}$};
  \node[fill = blue, text opacity = 1, opacity = 0.2]  (v2-2) at (2,-1) {$v_{2,2}$};
  \node[fill = blue, text opacity = 1, opacity = 0.2]  (v3-2) at (2,1) {$v_{3,2}$};
  \node[fill = red, text opacity = 1, opacity = 0.2]  (v4-2) at (2,2) {$v_{4,2}$};

  \draw (v1-2)  -- (v2-2) -- (w2) -- (v3-2) -- (v4-2);
 
  \node[fill = red, text opacity = 1, opacity = 0.2]  (v1-3) at (4,-2) {$v_{1,3}$};
  \node[fill = blue, text opacity = 1, opacity = 0.2]  (v2-3) at (4,-1) {$v_{2,3}$};
  \node[fill = blue, text opacity = 1, opacity = 0.2]  (v3-3) at (4,1) {$v_{3,3}$};
  \node[fill = red, text opacity = 1, opacity = 0.2]  (v4-3) at (4,2) {$v_{4,3}$};

  \draw (v1-3)  -- (v2-3) -- (w4) -- (v3-3) -- (v4-3);

  \node[fill = blue, text opacity = 1, opacity = 0.2] (w) at (6,-1) {$v_{1,4}$};
    \draw (w6) -- (w) ;
\end{tikzpicture}
\caption{Skeleton of the clause gadget for $d=2$ (on the left) and $d=3$ (on the right) before replacing vertices with cliques.}
\label{fig:clause_gadget}
\end{figure}
\smallskip

We now build a graph $G_{\phi,d}$ as follows :

\begin{enumerate}
    \item \emph{Mapping clauses to grid vertices.}  
    Choose an injective function $\textsf{cell} : [m] \to V(G_{d,p}) \setminus \{o\}$  mapping each clause of $\phi$ to a distinct non-central vertex of $G_{d,p}$. Call $\textsf{clause}=\textsf{cell}^{-1}$ so that, when defined, $\textsf{clause}(t)$  is the index of the clause associated to the node $t\in V(G_{d,p})$.

    \item \emph{Literal sets along the BFS tree.}  
    For each $t \in V(G_{d,p})$, let $g(t)$ be the set of literals appearing in clauses $\textsf{clause}(t')$ for all descendants $t'$ of $t$ in $T_{d,p}$.  
    Fix an arbitrary ordering of $g(t)$:
    $$
        g(t) = \{\ell_{t,1}, \dots, \ell_{t,|g(t)|}\}.
    $$

    \item \emph{Placing gadgets.}  
    For each $t \in V(G_{d,p})$, create a copy $H_t$ of $H_{g(t),d}$. To avoid heavy notation, we will always use \emph{the vertex $v_{i,j,\ell}$ of $H_t$} to talk about the copy of the vertex $v_{i,j,\ell}$ which appears in the gadget $H_t$. If $\textsf{clause}(t)$  is defined, let $\ell_1, \ell_2, \ell_3$ be its three literals and add a path $a_t^1 a_t^2 a_t^3$ with edges:
    $$
        a_t^1 v_{1,d+1,\ell_1},\quad
        a_t^2 v_{1,d+1,\ell_2},\quad
        a_t^3 v_{1,d+1,\ell_3}.
    $$

    \item \emph{Propagation edges.}  
    For each edge $t_1 t_2 \in E(T_{d,p})$ of dimension $d'$ and for each $\ell \in g(t_1) \cap g(t_2)$, if $d'\neq 1$ (resp. $d'=1$), add the edge between $v_{4,d',\ell}$ (resp. $v_{2d+1,1,\ell}$) of $H_{t_1}$ and $v_{1,d',\ell}$ of  $H_{t_2}$.


    \item \emph{Variable-consistency gadget.}  
    Add a clique $U = \{u_1, \dots, u_n\}$ and connect $u_i$ to the vertex $v_{1,d+1,\ell}$ of $H_o$ whenever $x_i = \ell$.
\end{enumerate}

By construction of $G_{\phi,d}$, any literal $\ell$ belongs to at most $h$ sets $g(t)$, where $h$ is the height of $T_{d,p}$ (at most $pd/2$).  
Each $H_t$ contains $O(d)$ vertices for each literal in $g(t)$.  
Since $p \leqslant n^{1/d} + 2$, the total number of vertices in $G_{\phi,d}$ is
$$
    O\left(d^2\, n^{\frac{d+1}{d}}\right).
$$

We now prove that $G_{\phi,d}$ admits a $2$-subcoloring if and only if there exists a truth assignment such that each clause contains at least one literal set to \textsf{true}and at least one literal set to \textsf{false}.

We begin with a simple claim, which is a slight generalization of one appearing in~\cite{marin2025subcoloring}, and we reprove it here for completeness.

\begin{claim}\label{claim:PropagationSubcoloring}
Let $G$ be a graph such that $V(G)$ admits a bipartition $(A,B)$ with $|A|,|B|\geqslant 3$ and 
$$
E(G) = \{aa' : a,a'\in A\}\cup \{bb' : b,b'\in B\}\cup M,
$$
where $M\subseteq A\times B$ is a matching with at least two edges. Then, for any $2$-subcoloring $\mu : V(G) \rightarrow \{0,1\}$ of $G$ and any edge $ab\in M$, we must have $\mu(a) \neq \mu(b)$.
\end{claim}

\begin{claimproof}
Suppose, for the sake of contradiction, that there exists an edge $ab\in M$ with $\mu(a) = \mu(b)$. Without loss of generality, assume $\mu(a)=0$. Then, for any $b'\in B\setminus \{b\}$, it must hold that $\mu(b')=1$, otherwise $abb'$ would form a monochromatic $P_3$ of color $0$. Similarly, for any $a'\in A\setminus \{a\}$, we have $\mu(a')=1$. Now take $a'b'\in M\setminus \{ab\}$ and any vertex $a''\in A\setminus \{a,a'\}$. Then $a''a'b'$ is a monochromatic $P_3$ of color $1$, a contradiction.
\end{claimproof}

Suppose that $G_{\phi,d}$ has a $2$-subcoloring $\mu : V(G)\rightarrow \{0,1\}$. We define the following truth assignment: for $1\leqslant i \leqslant n$, set $x_i$ to \textsf{true} if and only if $\mu(u_i)=1$. 

Now observe that, for any literal $\ell=x_i$, the color assigned by $\mu$ to the vertex $v_{1,d+1,\ell}$ of $H_o$ is $1-\mu(u_i)$. The proof is similar to the one of Claim~\ref{claim:PropagationSubcoloring}. Suppose that $\mu(u_i)=\mu(v_{1,d+1,\ell})=0$. Then, for any $j\neq i$, observe that $\mu(u_j)=1$, otherwise it would create a path on three vertices $u_ju_iv_{1,d+1,\ell}$ of color $0$. Similarly, for any $\ell\in \mathcal{L}$ such that $\ell\neq x_i$, $\mu(v_{1,d+1,\ell})=1$. Since $|\mathcal{L}|\geqslant 3$, take two different variable $x_{j_1}$ and $x_{j_2}$, both different from $i$, and consider $\ell_1\in L$ such that $\ell_1=x_{j_1}$. We have $\mu(v_{1,d+1,\ell_1})=\mu(u_{j_1})=\mu(u_{j_2})=1$, while $v_{1,d+1,\ell_1}u_{j_1}u_{j_2}$ is a path on three vertices. It is a contradiction.

We call $C_o$ the clique $C_{1,d+1}$ in the clause gadget $H_o$ of the root $o$. Then, we show that that the colors of $H_o$ propagates in a unique way to all other cliques of the graph $G_{\phi,d}$. Observe that for any clique $C$ in $H_t$ for some $t\in V(G_{d,p})$ such that $g(t)$ is not empty, there is a sequence of cliques of size at least $C_1,...,C_k$ such that $C_1=C_o$ and $C_k=C$ and such that for any $1\leqslant i <k$, there is a matching between $C_i$ and $C_{i+1}$. By successive applications of Claim~\ref{claim:PropagationSubcoloring}, we obtain that for any pair of literals $\ell,\ell' \in g(t)$, the colors assigned by $\mu$ to the vertices $v_{i,j,\ell}$ and $v_{i,j,\ell'}$ of $H_t$ are equal if and only if the colors assigned to $v_{1,d+1,\ell}$ and $v_{1,d+1,\ell'}$ in $H_o$ are the same, which in turn holds if and only if $\mu(u_i)=\mu(u_{i'})$ for $i,i'\in [n]$ such that $\ell=x_i$ and $\ell'=x_{i'}$. 

Now, let $\ell_1,\ell_2,\ell_3$ be the three literals appearing in the $j$th clause. Suppose, by contradiction, that under the truth assignment, the three corresponding variables are either all \textsf{true} or all \textsf{false}. Then, by the argument above, the three vertices $v_{1,2d,\ell_1}$, $v_{1,2d,\ell_2}$, and $v_{1,2d,\ell_3}$ of $H_t$ all have the same color under $\mu$, say color $0$. Consider the vertex $a_t^1$: it must have color $1$, otherwise $a_t^1v_{1,d+1,\ell_1}v_{1,d+1,\ell_2}$ would form a monochromatic $P_3$ of color $0$. Similarly, both $a_t^2$ and $a_t^3$ must have color $1$, which creates a monochromatic $P_3$ of the form $a_t^1a_t^2a_t^3$, a contradiction. 

Therefore, the formula $\phi$ admits a truth assignment in which each clause contains at least one \textsf{true} variable and at least one \textsf{false} variable.

\smallskip

Assume $\phi$ admits a truth assignment $\tau:\{x_1,\dots,x_n\}\to\{0,1\}$ such that in each clause at least one literal is \textsf{true} and at least one literal is \textsf{false} (we use $1=\textsf{true}$ and $0=\textsf{false}$). 
We construct a 2-subcoloring $\mu:V(G_{\phi,d})\to\{0,1\}$ as follows.

For each variable $x_i$, let $\mu(u_i):=\tau(x_i)$.  
For every literal $\ell=x_i$ that appears in the center gadget $H_o$, set
$$
\mu\big(v_{1,d+1,\ell}\big) := 1-\mu(u_i) = 1-\tau(x_i).
$$
for the vertex $v_{1,d+1,\ell}$ of $H_o$. Color every clique of $H_o$, following the propagation rule of Claim~\ref{claim:PropagationSubcoloring}. Similarly, along the BFS tree $T_{d,p}$, propagate the 2-subcoloring to every clique of all gadgets $H_t$ for $t\in V(G_{d,p})$ according to Claim~\ref{claim:PropagationSubcoloring}.  

Finally, for each node $t\in V(G_{d,p})$ such that $j=\textsf{clause}(t)$ is defined, let $\ell_1,\ell_2,\ell_3$ be the literals of the $j$th clause and set 
$$
\mu(a_t^b) = 1-\mu\big(v_{1,d+1,\ell_b}\big), \quad b\in\{1,2,3\}.
$$
It can be observed that there is no monochromatic $P_3$: indeed, for any edge $uv\in E(G_{\phi,d})$ with $u$ and $v$ belonging to different cliques (either in the same cell gadget or in two adjacent ones), we have $\mu(u)\neq \mu(v)$. Thus, the only possible monochromatic $P_3$ would be of the form $a_t^1a_t^2a_t^3$ for some cell $t\in V(G_{d,p})$, which is not monochromatic by the hypothesis on the truth assignment.

\subparagraph{Representation when $d=2$.} 
We describe a representation of $G_{\phi,2}$ using similarly sized fat objects. Define a bijective function $h : \mathcal{L} \rightarrow [4n]$ satisfying the condition that for any $\ell \in \mathcal{L}$, if $\ell$ corresponds to variable $x_i$ for some $i \in [n]$, then $h(\ell) \in \{4i, 4i+1, 4i+2, 4i+3\}$. 
Note that the literals associated with the same variable are consecutive in the order induced by $h$.

Let $\varepsilon > 0$ be a small real number to be fixed later. 
We first describe how to embed the cell gadget $H_{L,2}$ into a $10 \times 10$ rectangle. 
The corresponding representation is illustrated in Figure~\ref{fig:gadget2D}, and we now give a formal definition. 
For each $\ell \in L$, we create ten polygons as follows:
\begin{itemize}
    \item For each $x \in \{0,4,8\}$, add the polygon with vertices $(x, h(\ell)\varepsilon)$, $(x+1,1)$, $(x+2, h(\ell)\varepsilon)$, and $(x+1,-1)$. These polygons represent $v_{1,1,\ell}$, $v_{3,1,\ell}$ and $v_{5,1,\ell}$, respectively.
    \item For each $x \in \{2,6\}$, add the polygon with vertices $(x, h(\ell)\varepsilon)$, $(x+1 + h(\ell)\varepsilon, 1)$, $(x+2, h(\ell)\varepsilon)$, and $(x+1 + h(\ell)\varepsilon, -1)$. These polygons represent $v_{2,1,\ell}$ and $v_{4,1,\ell}$, respectively.
    \item For each $y \in \{-5, -3, 1, 3\}$, add the polygon with vertices $(3 + h(\ell)\varepsilon, y)$, $(2, y+1)$, $(3 + h(\ell)\varepsilon, y+2)$, and $(4, y+1)$. These polygons represent $v_{1,2,\ell}$, $v_{2,2,\ell}$, $v_{3,2,\ell}$ and $v_{4,2,\ell}$, respectively.
    \item Finally, add the polygon with vertices $(7 + h(\ell)\varepsilon, -3)$, $(6, -2)$, $(7 + h(\ell)\varepsilon, -1)$, and $(8, -2)$. This polygon represents the vertex $v_{1,3,\ell}$.
\end{itemize}

One can observe that two polygons intersect if and only if the corresponding vertices of $H_{L,2}$ are adjacent. 
Moreover, any two vertices of $H_{L,2}$ that do not belong to the same clique in the gadget intersect only at a single point. 
Later, in the construction of the full representation of $G_{\phi,2}$, some of these single-point intersections may be removed depending on the edges between cells in the tree $T_{p,2}$.

For any cell $t \in G_{2,p}$, we construct the corresponding gadget $H_{g(t),2}$ using the polygons described above. 
Additionally, if $\textsf{clause}(t)$ is defined, we add three polygons representing the vertices $a_t^1, a_t^2,$ and $a_t^3$ . 
Let $\ell_1, \ell_2, \ell_3$ denote the three literals of $\textsf{clause}(t)$, and assume that $h(\ell_1) < h(\ell_2) < h(\ell_3)$. 
Then we add:
\begin{itemize}
    \item the triangle with vertices $(7 + h(\ell_1)\varepsilon, -3)$, $(7 + h(\ell_1)\varepsilon, -5)$, and $(6, -4)$ to represent $a_t^1$;
    \item the polygon with vertices $(7 + h(\ell_2)\varepsilon, -3)$, $(8, -4)$, $(7 + h(\ell_2)\varepsilon, -5)$, and $(6, -4)$ to represent $a_t^2$;
    \item the triangle with vertices $(7 + h(\ell_3)\varepsilon, -3)$, $(7 + h(\ell_3)\varepsilon, -5)$, and $(8, -4)$ to represent $a_t^3$.
\end{itemize}

It can be observed that $a_t^1 a_t^2 a_t^3$  forms a path on three vertices in the polygonal representation. 
We denote by $\mathcal{P}_t$ the set of polygons obtained, which fit exactly into a $10 \times 10$ square. 
The full representation of $G_{\phi,2}$ is constructed as follows:
\begin{enumerate}
    \item For each node $t = (t_1, t_2) \in G_{2,p}$, place a copy of $\mathcal{P}_t$ within the square with vertices 
    $(10t_1, 10t_2)$, $(10t_1, 10(t_2 + 1))$, $(10(t_1 + 1), 10(t_2 + 1))$, and $(10(t_1 + 1), 10t_2)$.
    \item For each variable $x_i$, define
    $$
        h_i = \min \{ h(\ell) \mid \ell \in L, \ \ell = x_i \}.
    $$
    Add a polygon with vertices $(7 + h_i\varepsilon, -3)$, $(7 + (h_i + 3)\varepsilon, -3)$, $(8, -4)$, $(7, -5)$, and $(6, -4)$ to represent $u_i$, where coordinates are taken inside the $10\times10$ square which contains $H_o$, using the same origin as in the definition of the polygons.
    This polygon intersects all polygons representing vertices $v_{1, d+1, \ell}$ whenever $\ell = x_i$, since $h(\ell) \in \{h_i, \dots, h_i + 3\}$.
    \item For any edge $(t_1, t_2)(t_1 + 1, t_2) \in E(G_{2,p}) \setminus E(T_{2,p})$ (i.e., an edge of the grid not belonging to the tree), remove every point of every polygon that lies on the segment between the points $(10(t_1 + 1), 10t_2)$ and $(10(t_1 + 1), 10(t_2 + 1))$.
    Similarly, for any edge $(t_1, t_2)(t_1, t_2 + 1) \in E(G_{2,p}) \setminus E(T_{2,p})$, remove the segment between $(10t_1, 10(t_2 + 1))$ and $(10(t_1 + 1), 10(t_2 + 1))$.
    This ensures that the intersections between polygon sets corresponding to different grid nodes correspond exactly to the \emph{propagation edges} of $G_{\phi,2}$. 
    Note that removing these points does not affect the fact that the objects remain similarly sized fat objects. 
    Alternatively, one could avoid deletions by slightly shifting the polygons to eliminate any undesired intersections.
\end{enumerate}

It remains to fix the value of $\varepsilon$. 
The only requirement is that $h(\ell)\varepsilon$ must be small enough for any $\ell \in L$ so as not to create unintended intersections between polygons. 
This can be achieved, for instance, by taking $\varepsilon = \frac{1}{2|L|}$.

This completes the construction of the representation of $G_{\phi,2}$ as the intersection graph of similarly sized objects in dimension~$2$.

\begin{figure}[t]
    \centering
\begin{tikzpicture}[scale=1.1]
\foreach \x in {-1,0,1,2,3,4,5,6,7,8, 9}{
    \draw[black!20] (\x, -5) -- (\x,5); 
}
\foreach \y in {-5,...,5}{
    \draw[black!20] (-1,\y) -- (9,\y); 
}

\foreach \x in {0,4,8}{
    \fill[draw, black, fill = green, fill opacity = 0.2] (\x-1,0) -- (\x,1) -- (\x+1,0) -- (\x, -1) -- (\x-1,0) ; 
    \fill[draw, black, fill = green, fill opacity = 0.2] (\x-1,0.15) -- (\x,1) -- (\x+1,0.15) -- (\x, -1) -- (\x-1,0.15) ;     
    \fill[draw, black, fill = green, fill opacity = 0.2] (\x-1,0.3) -- (\x,1) -- (\x+1,0.3) -- (\x, -1) -- (\x-1,0.3) ;     
   \fill[draw, black, fill = red, fill opacity = 0.2] (\x-1,0.45) -- (\x,1) -- (\x+1,0.45) -- (\x, -1) -- (\x-1,0.45) ; }
    
\foreach \x in {0,4}{
\fill[draw, black, fill = red, fill opacity = 0.2] (\x+1,0) -- (\x+2,1) -- (\x+3,0) -- (\x+2,-1) -- (\x+1,0);
\fill[draw, black, fill = red, fill opacity = 0.2] (\x+1,0.15) -- (\x+2+0.15,1) -- (\x+3,0.15) -- (\x+2+0.15,-1) -- (\x+1,0.15);
\fill[draw, black, fill = red, fill opacity = 0.2] (\x+1,0.3) -- (\x+2+0.3,1) -- (\x+3,0.3) -- (\x+2+0.3,-1) -- (\x+1,0.3);
\fill[draw, black, fill = green, fill opacity = 0.2] (\x+1,0.45) -- (\x+2+0.45,1) -- (\x+3,0.45) -- (\x+2+0.45,-1) -- (\x+1,0.45);
 } 
 
 \foreach \x/\y in {0/0, 0/-4, 4/-4}{ 
\fill[draw, black, fill = green, fill opacity = 0.2] (\x+1,\y+2) -- (\x+2,\y+1) -- (\x+3,\y+2) -- (\x+2,\y+3) -- (\x+1,\y+2);
\fill[draw, black, fill = green, fill opacity = 0.2] (\x+1,\y+2) -- (\x+2+0.15,\y+1) -- (\x+3,\y+2) -- (\x+2+0.15,\y+3) -- (\x+1,\y+2);
\fill[draw, black, fill = green, fill opacity = 0.2] (\x+1,\y+2) -- (\x+2+0.3,\y+1) -- (\x+3,\y+2) -- (\x+2+0.3,\y+3) -- (\x+1,\y+2);
\fill[draw, black, fill = red, fill opacity = 0.2] (\x+1,\y+2) -- (\x+2+0.45,\y+1) -- (\x+3,\y+2) -- (\x+2+0.45,\y+3) -- (\x+1,\y+2);
}

 \foreach \x/\y in {0/2, 0/-6}{ 
\fill[draw, black, fill = red, fill opacity = 0.2] (\x+1,\y+2) -- (\x+2,\y+1) -- (\x+3,\y+2) -- (\x+2,\y+3) -- (\x+1,\y+2);
\fill[draw, black, fill = red, fill opacity = 0.2] (\x+1,\y+2) -- (\x+2+0.15,\y+1) -- (\x+3,\y+2) -- (\x+2+0.15,\y+3) -- (\x+1,\y+2);
\fill[draw, black, fill = red, fill opacity = 0.2] (\x+1,\y+2) -- (\x+2+0.3,\y+1) -- (\x+3,\y+2) -- (\x+2+0.3,\y+3) -- (\x+1,\y+2);
\fill[draw, black, fill = green, fill opacity = 0.2] (\x+1,\y+2) -- (\x+2+0.45,\y+1) -- (\x+3,\y+2) -- (\x+2+0.45,\y+3) -- (\x+1,\y+2);
}

\foreach \x/\y in {4/-6}{
\fill[draw, black, fill = red, fill opacity = 0.2] (\x+1,\y+2) -- (\x+2,\y+1) --  (\x+2,\y+3) -- (\x+1,\y+2);
\fill[draw, black, fill = red, fill opacity = 0.2] (\x+1,\y+2) -- (\x+2+0.3,\y+1) -- (\x+3,\y+2) -- (\x+2+0.3,\y+3) -- (\x+1,\y+2);
\fill[draw, black, fill = green, fill opacity = 0.2] (\x+2.45,\y+1) -- (\x+3,\y+2) -- (\x+2.45,\y+3) -- (\x+2.45,\y+1);
}

\end{tikzpicture}
    \caption{2D representation of $H_{L,2}$ with similarly sized fat objects.}
    \label{fig:gadget2D}
\end{figure}
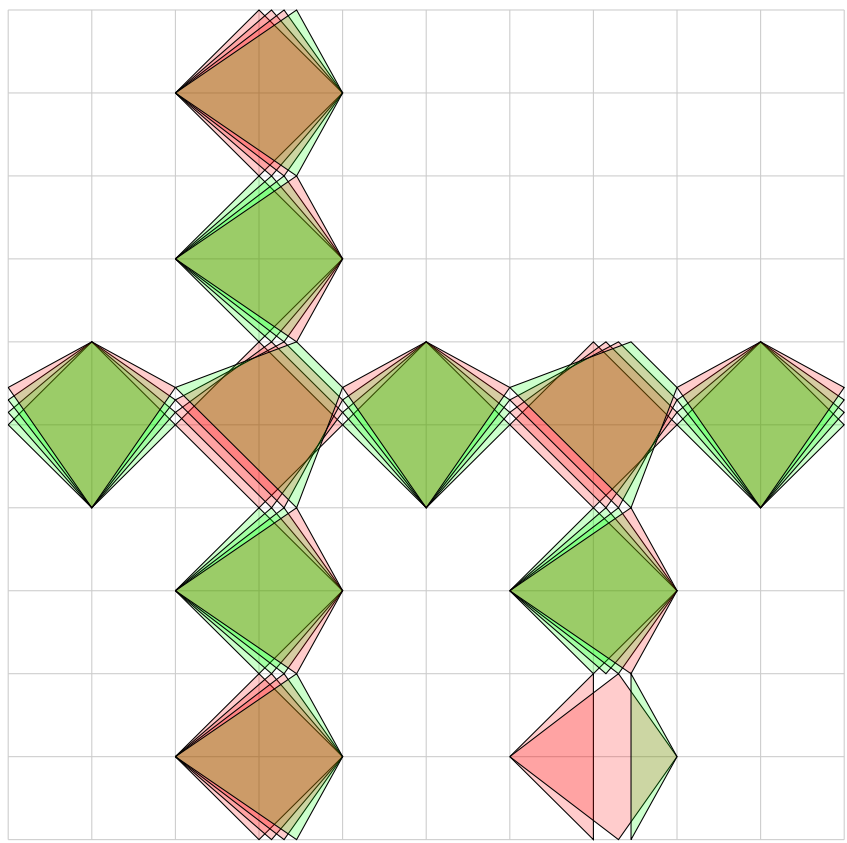

\subparagraph{Representation with unit balls when $d\geqslant 3$.} We now prove that for any $d\geqslant 3$, $G_{\phi,d}$ is a unit ball graph. Throughout, pick a sufficiently small $\varepsilon>0$ (specified below). We fix an ordering of the literals $h : \mathcal{L} \rightarrow [4n]$ in a similar way to the case $d=2$.

\medskip
Let $L\subseteq \mathcal{L}$ be a set of literals. The heart of the proof will be to represent $H_{L,d}$ with a $d$-dimensional box in a similar fashion than the $2$-dimensional case, except that the objects will be balls of radius $1/2$ this time. Recall that for every $j \in [d]$, $e_j$ denotes the $d$-dimensional unit vector with a $1$ in its $j$th coordinate.  We first embed one copy of $H_{L,d}$ inside an axis-aligned box of side lengths
$(2d+1)\times 5 \times \cdots \times 5 \times (3-\varepsilon m)$. 

We give the center of each ball representing a vertex of $H_{L,d}$. Let $\ell \in L$ be a literal. For the path $v_{1,1},\dots,v_{2d+1,1}$, place the center of the ball representing $v_{i,1,\ell}$ at $X_{i,1,\ell}:= ie_1+h(\ell)\varepsilon e_d$, where $1\leqslant  i\leqslant 2d+1$.

For each $2\leqslant j\leqslant d-1$, place four balls with centers at
\begin{itemize}
    \item $X_{2j-2,1,\ell}-2e_j$ to represent vertex $v_{1,j,\ell}$;
    \item $X_{2j-2,1,\ell}-1e_j$ to represent vertex $v_{2,j,\ell}$;
    \item $X_{2j-2,1,\ell}+1e_j$ to represent vertex $v_{3,j,\ell}$;
    \item $X_{2j-2,1,\ell}+2e_j$ to represent vertex $v_{4,j,\ell}$.
\end{itemize}

For $\ell \in L$, the ball representing $v_{1,d+1,\ell}$ is placed with center $X_{1,d+1,\ell} := 2de_{1}-e_{2}+\varepsilon h(\ell)e_{d}$.

\smallskip
One can observe that the ball representation respects the edges between corresponding vertices of $H_{L,d}$. It remains to place the balls of the last short path, \textit{i.e.} vertices $v_{i,d,\ell}$ for $i\in \{1,2,3,4\}$. 
The challenge is that this remaining dimension is already used to encode the offset $\varepsilon h(\ell)$. To overcome this issue, we will adopt a construction introduced in~\cite{de2020lower}.

\smallskip

To place those balls, we will work in the $3$-dimensional subspace
$$
S:=\{(x_1,x_2,0,\dots,0,x_d)\}\cong\mathbb{R}^3,
$$
and we write points as $(x_1,x_2,x_d)$, and use cylindrical coordinates centered at
$(2d-2,0,0)$ in the $(e_1,e_2)$-plane (\textit{i.e.} at the ball to which the path is attached). We have 
$$(x_1,x_2,x_d)\;\equiv\; (r,\theta,z)
$$
with
\begin{itemize}
    \item $r=\sqrt{(x_1-(2d-2))^2+x_2^2}$,
    \item $z=x_d$, 
    \item $\theta$ is the angle between $e_1$ and the projection of the point $(x_1,x_2,x_d)$ on the plane $(e_1,e_2)$.
\end{itemize}
Observe that the squared distance between two points $(1,\theta,z)$ and $(1,\theta',z')$ is exactly
\begin{equation}\label{eq:distCylindre}
(z-z')^2+2\sin^2\left(\frac{\theta-\theta'}{2}\right).
\end{equation}
Let $\varepsilon'>0$ be a constant yet to be fixed. Place the centers as follows :

\begin{itemize}
\item $X_{2,d,\ell}:=\ (1,\ -\tfrac{\pi}{2} +\varepsilon'h(\ell),\ \varepsilon h(\ell))$ to represent $v_{2,d,\ell}$ ;
\item $X_{1,d,\ell}:=\ (1,\ -\tfrac{\pi}{2} +\varepsilon'h(\ell),\ -1+\varepsilon m)$ to represent $v_{1,d,\ell}$ 
\item $X_{3,d,\ell}:=\ (1,\ \phantom{-}\tfrac{\pi}{2} + \varepsilon'h(\ell),\ \varepsilon h(\ell))$  to represent $v_{3,d,\ell}$ ;
\item $X_{4,d,\ell}:=\ (1,\ \phantom{-}\tfrac{\pi}{2} + \varepsilon'h(\ell),\ 1)$ to represent $v_{4,d,\ell}$.  
\end{itemize}
An illustration of how those points are placed for two consecutive literals in the order $h$ is given in Figure~\ref{fig:Cylinder}.

\begin{figure}[t]
    \centering
    \tdplotsetmaincoords{-10}{0}
\tdplotsetmaincoords{-10}{0}
\begin{tikzpicture}[scale=2, tdplot_main_coords]

\tdplotsetmaincoords{70}{0}

\begin{scope}[canvas is xz plane at y=1]
    \draw (0,0) circle (1);
\end{scope}

\begin{scope}[canvas is xz plane at y=-1]
    \draw (0,0) circle (1);
\end{scope}

\draw (1,-1,0) -- (1,1,0);
\draw (-1,-1,0) -- (-1,1,0);

\draw (0,1,0) -- (-0.389,1,0.921);
\draw (0,1,0) -- (-0.565,1,0.825);
\draw (0,-1,0) -- (0.389,-1,-0.921);
\draw (0,-1,0) -- (0.564,-1,-0.825);

\draw (0.389,0,-0.921) --(0,0,0) -- (-0.389,0,0.921);
\draw (0.564,0.05,-0.825) --(0,0.05,0) -- (-0.564,0.05,0.825);

\draw[->] (0,-1.5,0) -- (0,1.5,0);

\draw[dashed] (-0.389,0,0.921) -- (-0.389,1,0.921);
\draw[dashed]  (-0.564,0.05,0.825) -- (-0.564,1,0.825);

\draw[dashed] (0.389,0,-0.921) -- (0.389,-1,-0.921);
\draw[dashed]  (0.564,0.05,-0.825) -- (0.564,-1,-0.825);

\draw[red, <->] (0.1, -0.03, 0) -- (0.1, 0.08,0);
\node[] () at (0.2,0, 0) {\textcolor{red}{$\varepsilon$}};

\node[] () at (-0.46,0.9, 0.85) {\textcolor{red}{$\varepsilon'$}};

\node[circle, fill, blue, inner sep = 1] () at (-0.389,0,0.921) {};
\node[circle, fill, red, inner sep = 1] () at (-0.564,0.05,0.825) {};

\node[circle, fill, blue, inner sep = 1] () at (0,0,0) {};
\node[circle, fill, red, inner sep = 1] () at (0,0.05,0) {};

\node[circle, fill, blue, inner sep = 1] () at (0.389,0,-0.921) {};
\node[circle, fill, red, inner sep = 1] () at (0.564,0.05,-0.825) {};

\node[circle, fill, blue, inner sep = 1] () at (-0.389,1,0.921) {};
\node[circle, fill, red, inner sep = 1] () at (-0.565,1,0.825) {};

\node[circle, fill, blue, inner sep = 1] () at (0.389,-1,-0.921) {};
\node[circle, fill, red, inner sep = 1] () at (0.565,-1,-0.825) {};

\end{tikzpicture}
    \caption{An illustration of how are placed the centers of balls for two consecutive literals $\ell_1$ and $\ell_2$ in the $3$-dimensional subspace $S$. The two points along the central axis of the cylinder are $X_{2d-2,1,\ell_1}$ and $X_{2d-2,1,\ell_2}$ and are at distance exactly $\varepsilon$, as they are consecutive in the order $h$. }\label{fig:Cylinder}
\end{figure}
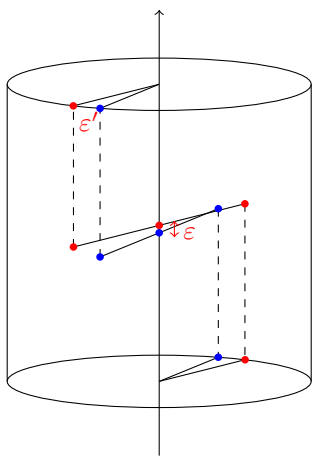
Following Equation~\ref{eq:distCylindre}, one can observe that the squared distance between the centers of $v_{1,d,\ell}$ and $v_{2,d,\ell}$ is $(1-\varepsilon(m-h(\ell)))^2$, and similarly for $v_{3,d,\ell}$ and $v_{4,d,\ell}$. This distance is less than $1$, under the condition that $\varepsilon \leqslant \frac{1}{m}$ (which will be satisfied once $\varepsilon$ is fixed).
For $\ell\neq \ell'$, the squared distance between the centers of $v_{1,d,\ell}$ and $v_{2,d,\ell'}$ is
\begin{equation}\label{eq:dist1}
(1-\varepsilon (m-h(\ell')))^2+2\sin^2\left(\frac{\varepsilon'}{2}(h(\ell)-h(\ell'))\right)
\end{equation}
Assuming that $\varepsilon$ and $\varepsilon'$ are smaller than $\frac{1}{m}$ (which will be the case later), the sum~(\ref{eq:dist1}) is at least
\begin{equation}\label{eq:dist2}
    (1-\varepsilon m)^2 + 2 \sin^2\left(\frac{\varepsilon'}{2}\right)
\end{equation}
We fix $\varepsilon = \frac{1}{4m^3}$ and $\varepsilon'= \frac{\pi}{m}$, and observe that the sum~(\ref{eq:dist2}) becomes
$$
\left(1-\frac{1}{4m^2}\right)^2 + 2 \sin^2\left(\frac{\pi}{2m}\right)\geqslant 1-\frac{1}{2m^2}+\frac{1}{16m^4}+2\left(\frac{2}{\pi} \frac{\pi}{2m}\right)^2>1
$$
\smallskip
The construction of the ball representation of $H_{d,L}$ is done. The case $d=3$ is illustrated in Figure~\ref{fig:clause_gadget3D}. Notice that it fits into a box of dimension $(2d+1)\times 5\times ...\times 5 \times (3-\varepsilon m)$. In addition, notice that for any $2\leqslant d'\leqslant d-1$, if we glue two ball representations $H_1$ and $H_2$ of $H_{d,L}$ along the face of the box corresponding to the $d'$-th dimension, then there is a matching between the clique $C_{1,d'}$ of $H_1$ and the clique $C_{2,d'}$ of $H_2$ (or vice-versa). A similar result holds for the first dimension : if we glue along the face of the first dimension of the box, then either there is a matching between the clique $C_{1,1}$ of $H_1$ and the clique $C_{2d+1,1}$ of $H_2$, or vice-versa. However, no such result holds for the $d$-th dimension (see Figure~\ref{fig:clause_gadget3D}). Thus we define two types of ball representations of $H_{L,d}$. Type I is precisely the one defined previously. Type II is the same, except that for any $\ell\in L$, the balls corresponding to $v_{1,d,\ell}, v_{2,d,\ell}, v_{3,d,\ell}$ and $v_{4,d,\ell}$ are placed differently. More precisely, place the centers as follows :
\begin{itemize}
\item $X_{2,d,\ell}:=\ (1,\ \tfrac{\pi}{2} +\varepsilon'h(\ell),\ \varepsilon h(\ell))$ to represent $v_{2,d,\ell}$ ;
\item $X_{1,d,\ell}:=\ (1,\ \tfrac{\pi}{2} +\varepsilon'h(\ell),\ -1+\varepsilon m)$ to represent $v_{1,d,\ell}$ 
\item $X_{3,d,\ell}:=\ (1,\ -\tfrac{\pi}{2} + \varepsilon'h(\ell),\ \varepsilon h(\ell))$  to represent $v_{3,d,\ell}$ ;
\item $X_{4,d,\ell}:=\ (1,\ -\tfrac{\pi}{2} + \varepsilon'h(\ell),\ 1)$ to represent $v_{4,d,\ell}$.  
\end{itemize}
Now, observe that if we glue a type I representation with a type II representation of two copies $H_1$ and $H_2$ of $H_{L,d}$ on the $d$-th face of the box, then there is a matching between the clique $C_{1,d}$ of $H_1$ and the clique $C_{4,d}$ of $H_2$ (or vice-versa).

\begin{figure}
    \centering
    \includegraphics[scale=0.4]{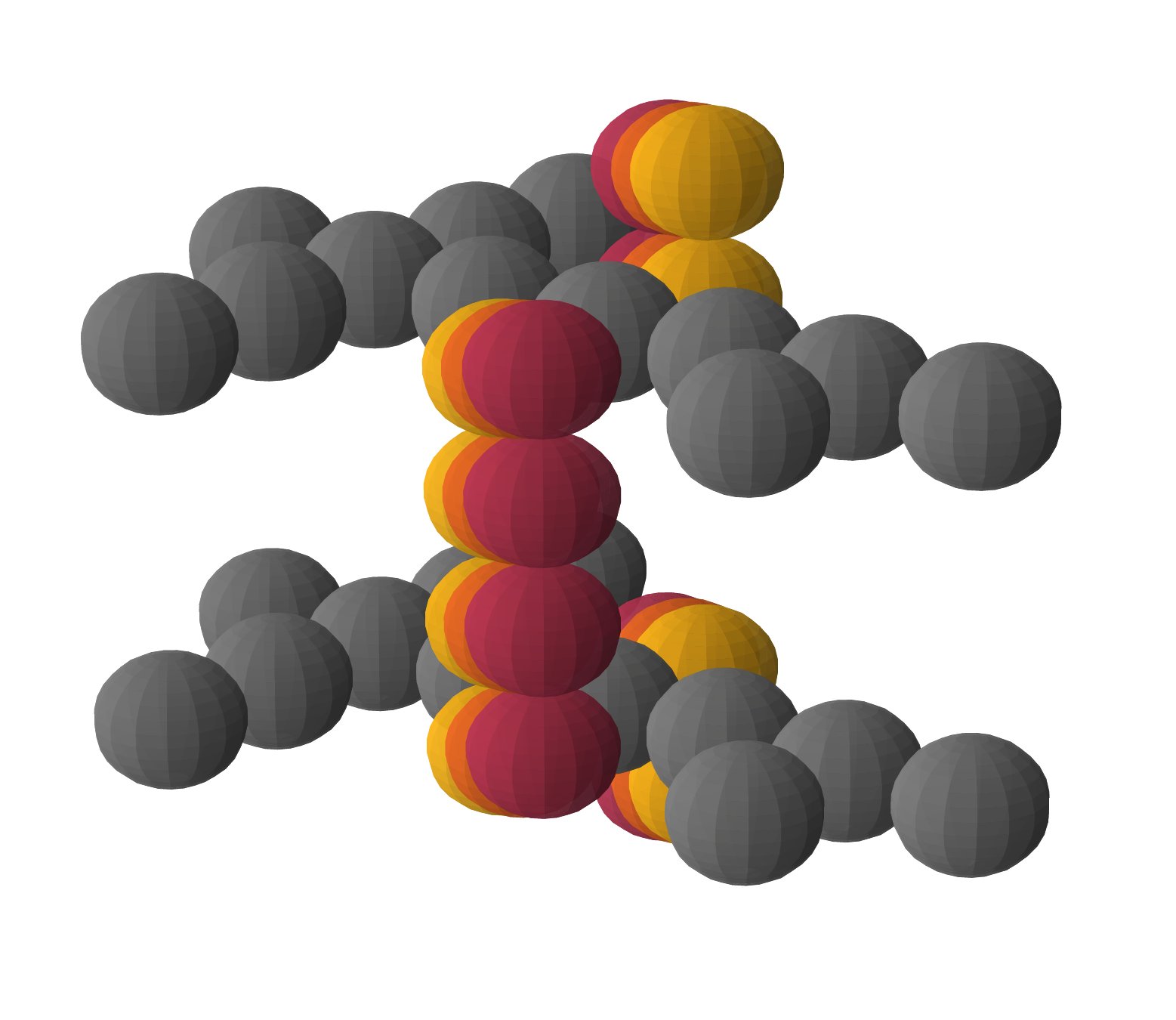}
    \caption{Representation of two adjacent clause gadgets with unit ball in $3$ dimensions. One clause gadget is of type I and the other of type II.}
    \label{fig:clause_gadget3D}
\end{figure}

\medskip
Let $t \in V(G_{d,p})$ be a vertex such that $\textsf{clause}(t)= \ell_1 \vee \ell_2 \vee \ell_3$ is defined. 
We describe how to add three balls to $H_{L,d}$ (whether it is of type~I or type~II) in order to represent the vertices $a_t^1$, $a_t^2$, and $a_t^3$. 
We use cylindrical coordinates, this time centered at $O := (2d,-1,0)$ instead of $(2d-2,0,0)$ as before. 
We place three balls with centers given by:
\begin{itemize}
    \item $A_{1}^t := \bigl(1,-\tfrac{3\pi}{10}, h(\ell_1)\varepsilon \bigr)$ to represent $a_t^1$;
    \item $A_{2}^t := \bigl(1,-\tfrac{4\pi}{10}, h(\ell_2)\varepsilon \bigr)$ to represent $a_t^2$;
    \item $A_{3}^t := \bigl(1,-\tfrac{6\pi}{10}, h(\ell_3)\varepsilon \bigr)$ to represent $a_t^3$.
\end{itemize}
One can observe that the distance between $A_{t}^1$ and $X_{1,d+1,\ell_1}$ is exactly $1$, and similarly for the pairs $(A_{t}^2, X_{1,d+1,\ell_2})$ and $(A_{t}^3, X_{1,d+1,\ell_3})$.
Moreover, provided that $\varepsilon$ is sufficiently small, we have $d(A_t^1, A_t^2) < 1$, $d(A_t^2, A_t^3) < 1$, and $d(A_t^1, A_t^3) > 1$, ensuring that the adjacencies among $a_t^1$, $a_t^2$, and $a_t^3$ in $G_{\Phi,d}$ are correctly represented. 
An illustration is given in Figure~\ref{fig:Checkgadget}. Finally, it can be observed that these three balls do not intersect any other ball in the gadget $H_{L,d}$ except those mentioned above.

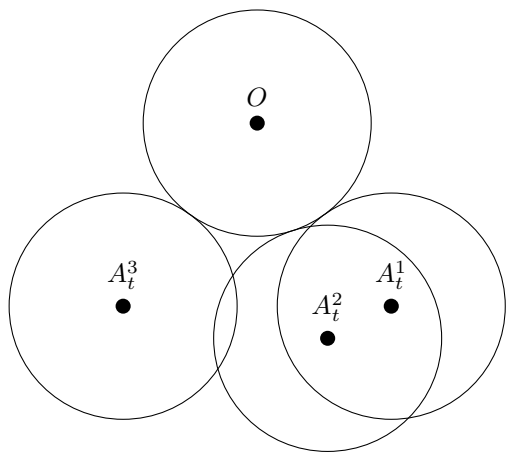
\begin{figure}[ht]
    \centering
    \begin{tikzpicture}[scale=3]
    \node[circle, fill, label = above:$O$, inner sep = 2] () at (0,0) {};
    \draw (0,0) circle (0.5);
    \node[circle, fill, label = above:$A_t^1$, inner sep = 2] () at (0.588,-0.809) {};
    \draw (0.588,-0.809) circle (0.5);
    \node[circle, fill, label = above:$A_t^2$, inner sep = 2] () at (0.309,-0.951) {};
    \draw (0.309,-0.951) circle (0.5);
    \node[circle, fill, label = above:$A_t^3$, inner sep = 2] () at (-0.588,-0.809) {};
    \draw (-0.588,-0.809) circle (0.5);
    \end{tikzpicture}
    \caption{An illustration of the positions of $O=(2d+2,-1,0)$, $A_t^1$, $A_t^2$, and $A_t^3$, projected onto the plane $(e_1,e_2)$.  There is a small offset along the $d$-th dimension of magnitude at most $m\varepsilon$, which is negligible.}\label{fig:Checkgadget}
\end{figure}
\medskip
Finally, we explain how to add the balls representing the vertices of the clique $U$ in the case $L=\mathcal{L}$. 
Recall that, in the bijective function $h$, the four occurrences of each literal appear consecutively. 
Given a variable $x_i$, define 
\[
\tilde{h}(i) = \frac{h(\ell_1)+h(\ell_2)+h(\ell_3)+h(\ell_4)}{4},
\]
where $\ell_1$, $\ell_2$, $\ell_3$, and $\ell_4$ are the four occurrences of the variable $x_i$. 
Note that if $h(\ell_1)= \min_{1\leqslant j \leqslant 4} h(\ell_j)$, then $\tilde{h}(i) = h(\ell_1)+\tfrac{3}{2}$. For each $1 \leqslant i \leqslant n$, we place a ball centered at 
\[
Y_i := \left(2d+2,\,-1-\sqrt{1-\left(\tfrac{2}{3}\varepsilon\right)^2},\,\tilde{h}(i)\varepsilon\right).
\]
Observe that $d(Y_i, X_{1,d+1,\ell_j}) \leqslant 1$ for each $j \in \{1,2,3,4\}$, and that for any literal $\ell \neq x_i$, we have $d(Y_i, X_{1,d+1,\ell}) > 1$. 
This construction is illustrated in Figure~\ref{fig:ConsistencyGadget}.
\begin{figure}[ht]
    \centering
    \begin{tikzpicture}[scale=0.5]
\draw[->] (0,-5) -- (0,9);
\node[] () at (-0.5,9.5) {$e_d$};
\node[circle, fill, label = left:$X_{1,d+1, \ell_1}$, inner sep=2] () at (0,0) {};
\node[circle, fill, label = left:$X_{1,d+1, \ell_2}$, inner sep=2] () at (0,1) {};
\node[circle, fill, label = left:$X_{1,d+1, \ell_3}$, inner sep=2] () at (0,2) {};
\node[circle, fill, label = left:$X_{1,d+1, \ell_4}$, inner sep=2] () at (0,3) {};

\node[circle, fill, label = left:$X_{1,d+1, \ell'}$, inner sep=2] () at (0,4) {};

\foreach \i in {0,1,2,3}{
	\draw (0,\i) circle (4) ; 
}
\draw[fill = blue, fill opacity = 0.2] (0,4) circle (4) ; 

\draw[red, fill = red, fill opacity = 0.2] (7.86, 1.5) circle (4);
\node[circle, fill = red, label = above:$Y_i$, inner sep = 2]  () at (7.86, 1.5) {};

\end{tikzpicture}
    \caption{Illustration of the placement of the points $Y_i$ for $1\leqslant i \leqslant n$ on the plane $(e_2,e_d)$. This point is at distance at most $1$ from the points $X_{1,d+1,\ell_1}, \ldots, X_{1,d+1,\ell_4}$, which are exactly the points such that $\ell_1=\ell_2=\ell_3=\ell_4=x_i$. For any other literal $\ell'$, the distance between $Y_i$ and $X_{1,d+1,\ell'}$ is strictly more than $1$.}
    \label{fig:ConsistencyGadget}
\end{figure}
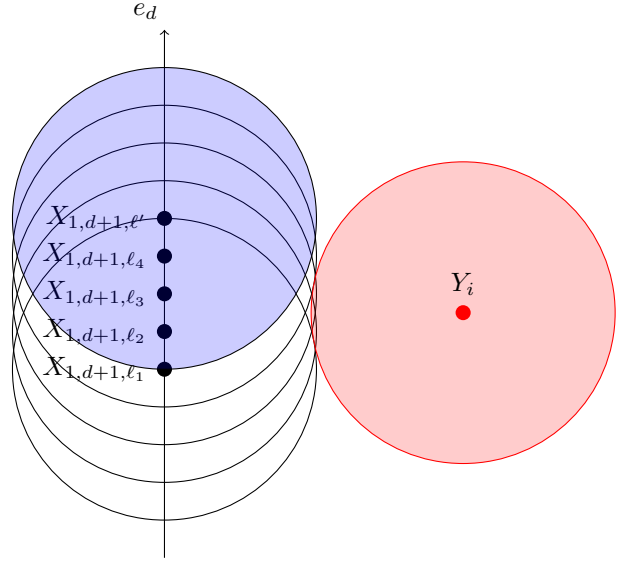

\medskip

It remains to complete the construction of $G_{\Phi,d}$. 
Consider a $d$-dimensional box of dimensions $(2d+1)p \times 5p \times \ldots \times 5p \times (3-\varepsilon m)p$, and partition it into $p^d$ sub-boxes of dimensions $(2d+1) \times 5 \times \ldots \times 5 \times (3-\varepsilon m)$. 
We map each of these sub-boxes to a vertex of $G_{d,p}$ in such a way that two sub-boxes share a common face if and only if the corresponding vertices are adjacent in $G_{d,p}$. 
Observe that $G_{d,p}$ is bipartite, and fix a proper $2$-coloring $\textsf{type}: V(G_{d,p}) \rightarrow \{I,II\}$. 
For each vertex $t \in V(G_{d,p})$, we place inside the corresponding sub-box a copy $H_t$ of the gadget $H_{d,g(t)}$ of type $\textsf{type}(t)$. Whenever $\textsf{clause}(t)$ is defined, we place balls to represent $a_t^1$, $a_t^2$ and $a_t^3$ as described above. Similarly, we add the balls representing the vertices of $U$ in the central sub-box of $H_o$ as described above. 

To complete the construction, it remains to remove the undesirable adjacencies that may occur along the face shared by two adjacent sub-boxes. 
Let $tt' \in E(G_{d,p}) \setminus E(T_{d,p})$, and let $d' \in \{1,\ldots,d\}$ be the dimension of this edge.  
Let $F$ be the common face between the sub-boxes associated with $t$ and $t'$. 
For every ball $B$ intersecting $F$, we shift $B$ slightly along dimension $d'$ so that its boundary lies at distance $\varepsilon''$ from $F$. 
Choosing $\varepsilon'' = \varepsilon/10$ is sufficient to ensure that adjacencies inside each clause gadget remain unchanged, while no ball intersects $F$ anymore. 
Repeating this operation for every such edge yields a valid ball representation of the graph induced by $G_{\Phi,d}$.

\end{proof}

\subsection{Lower Bound  for Two Sets Cut-Uncut}

We prove the same lower bound for the \textsc{Two Sets Cut--Uncut} problem. The construction is very similar to that of Section~\ref{sec:LowerBoundSubcoloring}, and therefore many geometric details are omitted.

\begin{theorem}\label{thm:LowerBoundCutUncut}
 \textsc{Two Sets Cut-Uncut} cannot be solved in time $2^{o\left(n^{1-1/(d+1)}\right)}$ on $n$-vertex intersection graphs of similarly sized of fat objects in $\mathbb{R}^d$ for any $d\geqslant 2$, and even in unit ball graphs in $\mathbb{R}^d$ for any $d\geqslant 3$.
\end{theorem}

\begin{proof}
Let $\Phi$ be a monotone $3$-SAT formula with $n$ variables and $m=4n$ clauses, where every clause contains exactly three positive literals and every variable appears in exactly four clauses.

We consider the graph $G'_{\Phi,d}$ obtained from the graph $G_{\Phi,d}$ used in the proof of Theorem~\ref{thm:LowerBound2Subcoloring}, with the following two modifications:
\begin{itemize}
  \item the vertex set $U$ is removed;
  \item for every node $t\in V(G_{d,p})$ the three vertices $a_t^1$, $a_t^2$ and $a_t^3$ now form a clique (instead of a $P_3$); we furthermore add two new vertices $b_t^1,b_t^2$ that are adjacent to each other and adjacent to all of $a_t^1,a_t^2,a_t^3$. The vertex $b_t^1$ is declared a terminal in $S$ and $b_t^2$ a terminal in $T$.
\end{itemize}

As in the original construction, every vertex of $G'_{\Phi,d}$ except the terminal vertices $b_t^1,b_t^2$ is associated to a unique literal of $\Phi$. As for $G_{\Phi,d}$, the gadgets are arranged so that all vertices that represent the same literal are connected within $G'_{\Phi,d}$. We do not repeat here the geometric embedding argument (intersection of fat polygons for $d=2$, unit balls for $d\geqslant 3$) since the modifications above do not affect that part of the proof.
\smallskip
Assign weight $1$ to an edge if and only if its two endpoints are both mapped to the \emph{same} literal; assign weight $0$ to every other edge. In particular, every edge incident to a terminal $b_t^1$ or $b_t^2$ has weight $0$ (because terminals are not mapped to literals).

\medskip
There exists a truth assignment of the variables of $\Phi$ such that no clause is all-\textsf{true} or all-\textsf{false} if and only if there exists a bipartition $(A,B)$ of $V(G'_{\Phi,d})$ with $S\subseteq A$, $T\subseteq B$, both $S$ and $T$ are connected in $G'_{\Phi,d}[A]$ and $G'_{\Phi,d}[B]$ respectively, and the total weight of edges between $A$ and $B$ equal to $0$.

We prove both directions.

\smallskip 
Let $\tau$ be a truth assignment of the variables of $\Phi$ in which every clause contains at least one \textsf{true} and at least one \textsf{false} literal. We construct a partition $(A,B)$ of $V(G'_{\Phi,d})$ as follows:
\begin{itemize}
  \item For every literal $\ell$, place all vertices of $G'_{\Phi,d}$ that are mapped to $\ell$ on the $A$-side if $\ell$ is assigned \textsf{true} by $\tau$, and on the $B$-side if $\ell$ is assigned \textsf{false}.
  \item Place every terminal $b_t^1\in S$ in $A$ and every $b_t^2\in T$ in $B$.
\end{itemize}

Because we place all vertices that represent the same literal on the same side, no edge of weight $1$ crosses the cut: every edge with weight $1$ has both endpoints representing the same literal, hence both endpoints lie in the same side. Therefore the cut weight is $0$.

It remains to argue that all terminals of $S$ are in the same connected component $G'_{\Phi,d}[A]$ and similarly for $T$ and $G'_{\Phi,d}[B]$. In fact, we will prove that both $G'_{\Phi,d}[A]$ and $G'_{\Phi,d}[B]$ are connected. By construction each literal's vertices form a connected subgraph, so the vertices grouped by literal are connected. Consider now any gadget corresponding to a node $t\in V(G_{d,p}')$; its three vertices $a_t^1,a_t^2,a_t^3$ represent the three literals of the corresponding clause. By the property of $\tau$ (no clause all-\textsf{true} or all-\textsf{false}), in each clause gadget at least one of the three $a_t^i$ lies in $A$ and at least one lies in $B$. Hence
\begin{itemize}
  \item the terminal $b_t^1\in A$ has a neighbor among the $a_t^i$ that is in $A$, so $b_t^1$ is not isolated in $G'[A]$; and
  \item the terminal $b_t^2\in B$ has a neighbor among the $a_t^i$ that is in $B$, so $b_t^2$ is not isolated in $G'[B]$.
\end{itemize}
Consequently $G'_{\Phi,d}[A]$ is connected and $G'_{\Phi,d}[B]$ is connected. Thus $(A,B)$ is a feasible partition with cut weight $0$.

\medskip
Conversely, suppose $(A,B)$ is a $S$-$T$-cut of $G'_{\Phi,d}$ of weight $0$ such that $S$ (resp.\ $T$) is contained in a connected component of $G'_{\Phi,d}[A]$ (resp. $G'_{\Phi,d}[B]$). We first observe:

\begin{enumerate}
  \item\label{itm:same-literal} Since every edge of weight $1$ joins two vertices mapped to the same literal, and the cut has weight $0$, no  edge of weight $1$ crosses the cut. Therefore every pair of adjacent vertices that are mapped to the same literal must lie in the same side. Because all vertices mapped to the same literal form a connected subgraph in the construction, it follows that \emph{all} vertices mapped to any fixed literal lie in the same side. Thus the partition induces a well-defined truth value for each literal (``side $A$'' or ``side $B$'').
  \item\label{itm:clause-not-all} Fix any clause gadget $t$. If all three vertices $a_t^1,a_t^2,a_t^3$ were in the same side, say all in $A$, then the terminal $b_t^2\in B$ would have no neighbor in $B$ (because $b_t^2$ is only adjacent to the $a_t^i$ and to $b_t^1$, and $b_t^1\in A$ by terminal constraint), contradicting the fact that $b_t^2$ is in the same connected components as all other terminals of $T$ in $G'_{\Phi,d}[B]$. Hence for every clause gadget $t$ the three vertices $a_t^1,a_t^2,a_t^3$ are \emph{not} all in $A$ and are \emph{not} all in $B$; equivalently, each clause gadget has at least one $a_t^i$ in $A$ and at least one in $B$.
\end{enumerate}

From (\ref{itm:same-literal}) we obtain a well-defined Boolean assignment $\tau$ to the variables: set a variable to \textsf{true} if the vertices representing its (positive) literal lie in $A$, and \textsf{false} if they lie in $B$. By (\ref{itm:clause-not-all}) every clause gadget contains at least one $a_t^i$ in $A$ and at least one in $B$, so each clause contains at least one \textsf{true} literal and at least one \textsf{false} literal under $\tau$. Therefore $\tau$ is a truth assignment in which no clause is all-\textsf{true} or all-\textsf{false}, as required.

This completes the equivalence and hence the proof.
\end{proof}

\section{Conclusion and further directions}


We introduced the \emph{weak square-root phenomenon}, revealing a new class of problems that admit $2^{\tilde O(n^{1-1/(d+1)})}$-time algorithms on geometric intersection graphs of similarly sized $\beta$-fat objects in $\mathbb{R}^d$, together with tight ETH lower bounds of $2^{\Omega(n^{1-1/(d+1)})}$, even on restricted subclasses. We established this phenomenon via a unified algorithmic and lower bound framework, and showed that it captures natural problems, including \textsc{2-Subcoloring} and \textsc{Two Sets Cut-Uncut}. Along the way, we introduced a new graph parameter called $\alpha$-modulator number, which may be of independent interest.
\medskip

We conclude with three open problems related to our lower-bound framework. 
First, the case of \emph{unit disk graphs} is not covered by our construction. 
Thus, we ask the following.

\begin{problem}
What is the complexity of \textsc{2-Subcoloring} and \textsc{Two Sets Cut--Uncut} on unit disk graphs?
\end{problem}

\begin{problem}
What is the complexity of \textsc{Two Sets Cut--Uncut} on intersection graphs of similarly sized $\beta$-fat objects in $\mathbb{R}^d$ in the unweighted case?
\end{problem}

Then, Theorem~\ref{thm:Separator} is not \emph{robust}, as it requires the geometric representation as part of the input, leading to the following problem.
\begin{problem}
Can we obtain a robust version of Theorem~\ref{thm:Separator} that does not rely on the geometric representation as input?
\end{problem}

We distinguish three possible outcomes for the first two problems: either they are polynomial-time solvable, or they admit a better subexponential algorithm than the one provided here, typically with running time $2^{O(\sqrt{n})}$, or it may be possible to successfully adapt the lower-bound framework to unit disk graphs.

Finally, we outline several directions for future research rather than concrete open problems:
\begin{itemize}
    \item Explore other problems that fit within this framework. In particular, it would be interesting to systematically consider the independence number $\alpha(G)$ as a natural parameter for FPT and XP algorithms. This direction has already received some attention~\cite{fomin2024path}.
    \item Investigate other classes of geometric graphs for which this framework could be applied.
    \item Continue the study of $\alpha$-modulator number in general graphs as a structural parameter.
\end{itemize}

\bibliography{biblio}

\end{document}